\documentclass[sigconf]{acmart}

\usepackage[english]{babel}
\usepackage{blindtext}
\usepackage{amsmath}
\usepackage{subcaption}
\usepackage[normalem]{ulem}
\usepackage{siunitx}
\usepackage{amsthm,amsmath}
\usepackage{graphicx}
\usepackage[most]{tcolorbox}
\usepackage[subtle]{savetrees}  
\newcommand{\ie}{\emph{i.e.,}\xspace}
\newcommand{\eg}{\emph{e.g.,}\xspace}
\newcommand{\etc}{\emph{etc.}\xspace}

\newcommand{\myitem}[1]{\vspace*{0.02in}\noindent\textbf{#1}}

\newtheoremstyle{sltheorem}
{}
{}
{\slshape}
{}
{\bfseries}
{.}
{ }
{}
\theoremstyle{sltheorem}
\newtheorem{theorem}{Theorem}

\tcolorboxenvironment{theorem}{
  enhanced,
  breakable,
  colback=gray!10,      % background
  colframe=gray!60,     % border
  boxrule=0.6pt,
  arc=2mm,
  left=1.2mm,
  right=1.2mm,
  top=1.0mm,
  bottom=1.0mm,
}

\makeatletter
\renewcommand{\proofname}{\bfseries Proof}
\renewenvironment{proof}[1][\proofname]{%
  \par\pushQED{\qed}%
  \normalfont \topsep6\p@\@plus6\p@\relax
  \trivlist
  \item[\hskip\labelsep\bfseries #1\@addpunct{.}]%
  \ignorespaces
}{%
  \popQED\endtrivlist\@endpefalse
}
\makeatother

\makeatletter
\newcommand{\mysubsubsection}{%
  \@startsection{subsubsection}{3}{\z@}%
    {2.0ex \@plus .2ex \@minus .2ex}% space above (tweak if you want)
    {1.0ex \@plus .2ex}% space below (positive => displayed heading, not run-in)
    {\bfseries}% font/style (subsection-like, but smaller)
}
\makeatother

\renewcommand\footnotetextcopyrightpermission[1]{} % removes footnote with conference info
\setcopyright{none}
\acmDOI{}

\acmISBN{}

\acmPrice{}

\begin{document}
\title{\emph{To Reconfigure or Not to Reconfigure:}\\ Optimizing All-to-All Collectives in Circuit-Switched Photonic Interconnects}

%\titlenote{Produces the permission block, and copyright information}
%\subtitle{Extended Abstract}

% \author{Paper \# 620, 12 pages}
\author{Anchengcheng Zhou}
% \authornote{Note}
% \orcid{1234-5678-9012}
\affiliation{%
  \institution{Princeton University}
  % \streetaddress{Address}
  % \city{City} 
  % \state{State} 
  % \postcode{Zipcode}
}
\email{ann.zhou@princeton.edu}

\author{Vamsi Addanki}
\affiliation{%
  \institution{Purdue University}
}
\email{vaddank@purdue.edu}

\author{Maria Apostolaki}
\affiliation{%
  \institution{Princeton University}
}
\email{apostolaki@princeton.edu}

% The default list of authors is too long for headers}
\renewcommand{\shortauthors}{X.et al.}

\sloppy

%!TEX ROOT=../paper.tex
\begin{abstract}

All-to-all collective communication is a core primitive in distributed machine learning and high-performance computing. At the server scale, the communication demands of these workloads are increasingly outstripping the bandwidth and energy limits of electrical interconnects, driving a growing interest in photonic interconnects.
However, leveraging these interconnects for all-to-all communication is nontrivial. %, as they rely on circuit switching with reconfiguration overheads.
The core challenge lies in jointly optimizing a sequence of topologies and flow schedules, reconfiguring only when the transmission savings from traversing shorter paths outweigh the reconfiguration cost.
Yet the search space of this joint optimization is enormous.
Existing work sidesteps this challenge by making unrealistic assumptions on reconfiguration costs so that it is never or always worthwhile to reconfigure.

In this paper, we show that any candidate sequence of topologies and flow schedules can be expressed as a sum of adjacency matrices and their powers.
This abstraction captures the entire solution space and yields a lower bound on all-to-all completion time. 
% Building on this formulation, we identify a family of topology sequences with strong symmetry and high expansion, which admits\maria{allows?} efficient schedules that can be generated via simple flow-splitting heuristics with low computational overhead.
% Building on this formulation, we identify a family of topology sequences with strong symmetry and high expansion, where bandwidth-efficient schedules exist, and which our efficient algorithm can generate. %via simple flow-splitting heuristics.
Building on this formulation, we identify a family of topology sequences with strong symmetry and high expansion that admits bandwidth-efficient schedules, which our algorithm constructs with low computational overhead.
Together, these insights allow us to efficiently construct near-optimal solutions, effectively avoiding enumeration of the combinatorial design space.
% multi-hop forwarding can be expressed via matrix powers, and use this observation to derive an abstraction that captures the entire solution space.
% This abstraction exposes the structure of the space and allows us to reason without enumerating individual solutions.
% In particular, we identify a high-performing region characterized by topology sequences exhibiting strong symmetry and high expansion; within this region, schedules can be constructed efficiently while achieving performance close to the global optimum.
Evaluation shows that our approach reduces all-to-all completion time by up to 44\% on average across a wide range of network parameters, message sizes and workload types.

\end{abstract}

\maketitle
\thispagestyle{plain}
\pagestyle{plain}

%!TEX ROOT=../paper.tex
\section{Introduction}

All-to-all is a backbone primitive in machine learning (ML) and high-performance computing (HPC), supporting multidimensional Fast Fourier Transform~\cite{10.1145/2555243.2555249,9652837} that is central to spectral solvers~\cite{CHATTERJEE201877,doi:10.1137/19M1303848}, imaging~\cite{9355267} and scientific simulations~\cite{doi:10.1137/11082748X}, and distributed training workloads \eg Deep Learning Recommendation Model~\cite{10.1145/3470496.3533727} and Mixture-of-Experts (MoE)~\cite{liao2025mixnet,NEURIPS2022_2f00ecd7,lepikhin2021gshard}. % involves repeated global transposes % operations \eg % that exchanges large embeddings %  that route tokens among all experts across GPUs
In many of these applications, the global data movement dominates the end-to-end completion time. For example, all-to-all communications account for 33\% to 55\% of the training iteration time in MoE~\cite{liao2025mixnet}.

% Yet all-to-all communication is particularly challenging, as its dense traffic pattern stresses the network globally and spawns many concurrent flows, easily creating hotspots and incasts~\cite{}.
Yet all-to-all communication is particularly challenging as the dense traffic and concurrent flows readily create hotspots and incasts.
To speed up these communications, data centers today increasingly deploy multi-GPU servers with packet-switched scale-up interconnects~\cite{nvidia2023superpod,10664412}, allowing a substantial fraction of communications to complete within a server.
However, packet-switched interconnects are increasingly constrained by hardware limits in \eg switch area, power needs and heat dissipation, and the slowing of bandwidth growth as the scaling by Moore's Law tapers off~\cite{10.1145/3387514.3406221}.

Circuit-switched (reconfigurable) photonic interconnects offer an attractive alternative: by moving switching into the optical domain,  
they deliver high bandwidth and energy efficiency. This has motivated many optical circuit switch designs~\cite{10.1145/3387514.3406221,10.1145/3651890.3672273,khani2021sipml,10.1145/3696348.3696856} that provide direct connections between communicating GPUs. %they avoid large packet crossbars, per-bit packet processing, and repeated SerDes overheads~\cite{}, % that vary in hardware and reconfigurability
However, these direct connections come at a cost: reconfiguration takes time, and per-reconfiguration delay varies from \si{\nano\second}--to--\si{\micro\second}~\cite{10.1145/3387514.3406221,10.1145/3651890.3672273,10.1145/3696348.3696856} to \si{\milli\second}~\cite{polatis} and \si{\minute}~\cite{zhao2025direct-connect}. This implies that at any moment, each host can use only a small number of optical circuits and must wait for reconfiguration to establish new ones. %Per-reconfiguration overheads vary by orders of magnitude from \si{\nano\second}–to-\si{\micro\second}~\cite{} to \si{\milli\second}~\cite{} and \si{\minute}~\cite{}.

In this paper, we seek to optimize all-to-all collective completion time in circuit-switched photonic interconnects by deciding when and how to reconfigure.
%, where performance hinges critically on finding a sequence of reconfigurations that minimizes the combined cost of data transmission and reconfiguration overhead.
%Deciding when and how to reconfigure requires a \emph{joint optimization} over all realizable topologies and flow schedules. % on the optical switch. 
%\maria{The prev two sentences are repetitive? }
Intuitively, only flows whose transmission savings can offset the reconfiguration cost should trigger an additional reconfiguration. But this decision cannot be made on a per-flow basis in isolation, because many flows share the same topology and interact.
Furthermore, the decisions on which topologies to use and what flows to run on these topologies are also tightly coupled: the chosen topology determines the feasible routes and their hop lengths, while the schedule (which flows use which routes in which stage) in turn determines which topologies are more useful. Even a small change in either can shift the best choice of the other. 
Since topology sequences and schedules are interdependent, we must jointly optimize topologies and schedules across all stages, rather than making independent, local decisions. While this joint optimization can be formulated, solving it by enumeration is intractable, as the search space spans all possible topology sequences, schedules, and even flow orderings, since all-to-all communication imposes no data dependencies that restrict ordering and results in combinatorial growth in the number of possible schedules. Moreover, this optimization must be repeated for each traffic instance and each reconfiguration overhead, rendering enumeration-based approaches prohibitively expensive.
%In fact, finding just a single optimal topology and its corresponding optical schedule is already NP-hard~\cite{foerster2018nphard} \ann{check}. 
Existing work \emph{reduces} the problem to make it tractable, either by assuming extreme reconfiguration delays that make reconfiguring always or never worthwhile~\cite{zhao2025direct-connect,10.1145/3748273.3749210}, or compressing along one dimension and optimizing only topology~\cite{liao2025mixnet} or collective schedule~\cite{zhao2025direct-connect,basu2024efficienta2a,liu2024teccl,cai2021sccl,shah2023taccl}. %, or restricts only to related but orthogonal settings~\cite{}. 
These reductions miss the key opportunity of \emph{reconfiguring during the collective}, which requires a joint optimization with realistic per-reconfiguration delay. %, leading to under-performance.

% Existing work \emph{reduces} the decision problem of when and how to reconfigure to make it tractable, but these simplifications lead to under-performance. 
% One line of work sidesteps the decision by making unrealistic assumptions on per-reconfiguration overhead: \cite{} assume the reconfiguration overhead to be effectively infinite so reconfiguring is never worthwhile, reducing the problem to static topology and routing design, while \cite{} assumes negligible reconfiguration overhead (or a fully connected fabric), reducing the problem to minimum Birkhoff-von Neumann decomposition. 
% Other approaches optimize only part of the stack.
% Collective synthesis either assumes a fixed topology and solves a partial optimization~\cite{}, or abstracts away the underlying topology~\cite{}.
% MixNet~\cite{liao2025mixnet} greedily constructs a single topology from a known traffic matrix whereas general direct-connect topology designs~\cite{} do not optimize for collective completion.
% Finally, several works are adjacent but orthogonal.
% TopoOpt~\cite{} and \cite{} exploit all-reduce-specific structure, hybrid data center designs~\cite{} fall back to packet switching for difficult traffic and often disallow multi-hop forwarding over optical circuits, and \cite{} assumes traffic is not fully transparent.

\begin{figure}[tp]
\centering
\includegraphics[width=0.95\columnwidth]{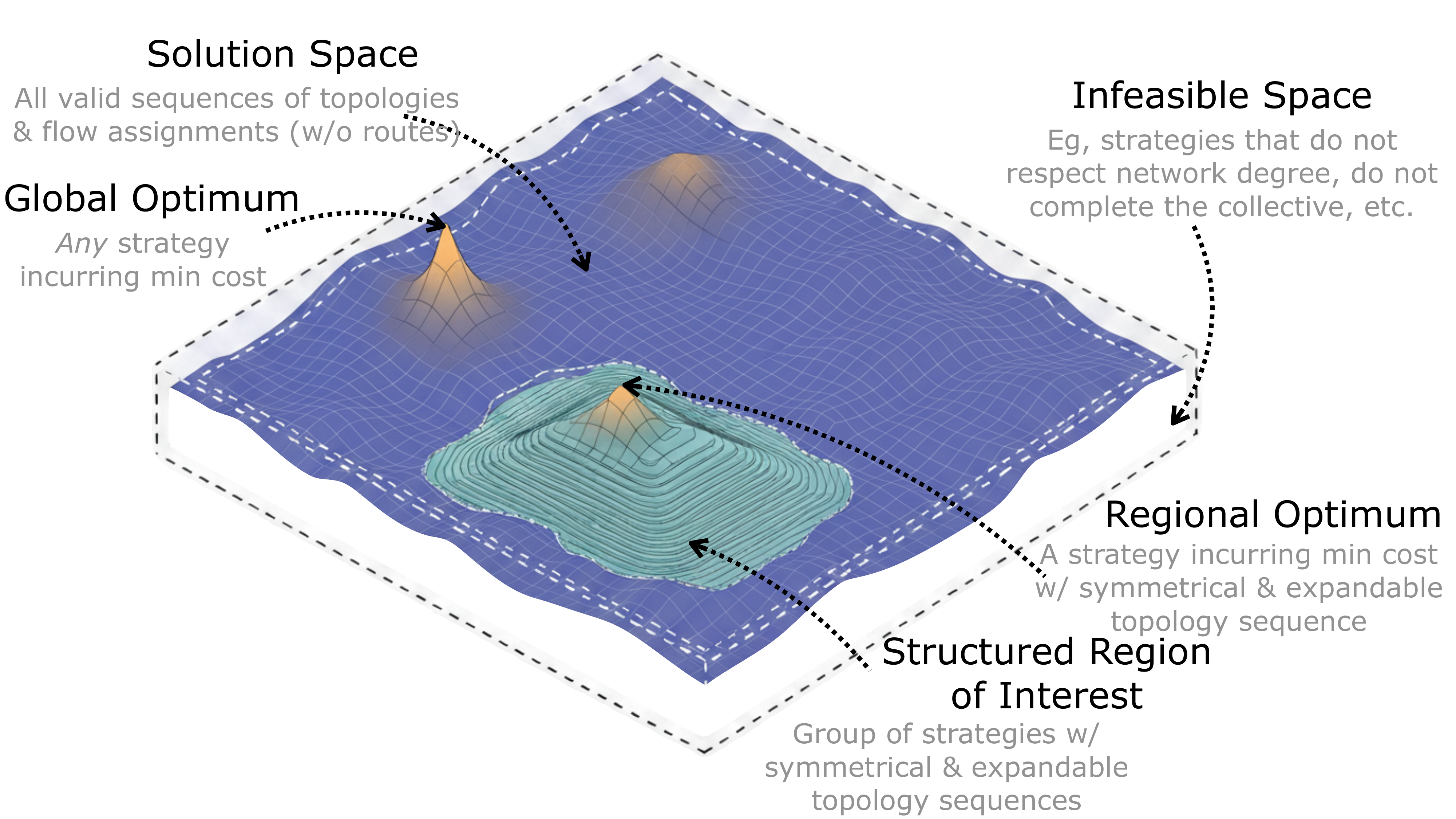}
\vspace{-0.1in}
\caption{Our approach maps the space of feasible solutions (strategies), \ie all possible sequences of topologies and flow assignments that map each flow to one of these topologies that satisfies the collective but without lower-level routing details, per network-number of reconfiguration instance,
%combinations of topology sequence topology-flow-mapping, routes,,,,), and 
identifies a structured region (characterized by topology sequences that exhibit strong symmetry and high expansion) in the space where the regional optimum is easy to find and close to the global optimum. Figure shown for illustration only (not to scale).
}
\label{fig:intro}
\vspace{-0.2in}
\end{figure}

%In this paper, we depart from the trajectory of existing work \maria{but we need to say how} and instead introduce an abstraction that captures the entire solution space of reconfiguration sequences, enabling reasoning on when and how to reconfigure without explicit enumeration. 
% \maria{Instead of following the trajectory of prior work that immediately searches for a concrete solution, we introduce an abstraction that maps the entire design space and allows us to reason about topology sequences without explicitly accounting for flow-level details that would require enumeration.}
Instead of following the trajectory of prior work that searches for concrete solutions, we introduce an abstraction that maps the entire space of feasible solutions to the all-to-all collective without explicitly modeling lower-level path details, thereby allowing us to reason about families of solutions without enumeration.
First, we show that the sequence of topologies and the flow schedules that will run on these topologies can be expressed as a sum of matrices by modeling multi-hop forwarding as matrix exponents. This matrix-based representation yields an explicit formulation for total collective completion time, 
%parameterized by the number of reconfigurations, 
cleanly exposing the tradeoff between transmission and reconfiguration costs, 
which we then use to derive a lower bound (\ie global optimum) on the completion time for all-to-all. %for any given set of network parameters. 
% \maria{the next two sentences might not be needed.}
% Moreover, this abstraction guides us to reason about regions of solutions together, abstracting away lower-level routing details.
% Next, we observe that the sequence of topologies underpins performance, as a well-chosen sequence simplifies the joint optimization to just constructing efficient schedules \ann{will this be confused with existing work?}. \maria{short answer: in your family of sequences you have an efficient algorithm to find the schedule? But should you not say this later? You could say your abstraction models all schedules of given topology sequence (like an abstart strategy). Then the bound says there is no schedule better than X, then you can find one of the concrete strategies that is X or close to X  with an efficient algorithm (i.e. not optimization) }
Next, leveraging the abstraction, we identify a region of the solution space characterized by topology sequences that exhibit strong intra- and inter-topology symmetry and high expansion. Within this region, (regionally) optimal schedules can be computed efficiently by a simple algorithm that combines flow grouping and shortest path routing. Furthermore, we show that the regional optimum is close to the global optimum.
Figure~\ref{fig:intro} summarizes the high-level idea of our approach. % that maps the entire solution space and then constrains to a structured region in the space where the regional optimum is easy to find and approximates the global optimum. % Also, note that our approach to generic to network parameters and workload types.

Evaluation shows that enabling reconfiguration during collectives reduces the total collective completion time by up to 44\% \footnote{An average taken based on Figures~\ref{fig:eval_main} and~\ref{fig:eval_main_size}.} on average across a wide range of network parameters, data sizes and workload types. We further validate our chosen region by demonstrating a small optimality gap between our solution and the global optimum. Furthermore, we provide quantitative evidence and intuition for when reconfiguration is beneficial.
The substantial reduction in all-to-all completion time can translate into significant cost savings for the workloads in today's ML and HPC where communication is a major contributor to overall expense.
%!TEX ROOT=../paper.tex
\section{Motivation}

We first define our problem setup (\S\ref{sec:motivation_setup}) and discuss 
% how existing work makes overly simplistic assumptions and under-performs
% vamsi: removed that, we don't want to make the reader angry if they had prior work in the area....
the limitations of existing approaches (\S\ref{sec:motivation_limitation}). Next, we use a motivating example to demonstrate the missed opportunity to \emph{selectively} reconfigure \emph{during} an all-to-all collective via jointly optimizing topologies and flow schedules. However, such joint optimization is prohibitively expensive (\S\ref{sec:motivation_example}), motivating the need for bandwidth-efficient all-to-all schedules that adapt to reconfiguration overheads.
\vspace{-0.15in}

\subsection{Problem Setup} \label{sec:motivation_setup}

\begin{figure}[tp]
\centering
\includegraphics[width=0.65\columnwidth]{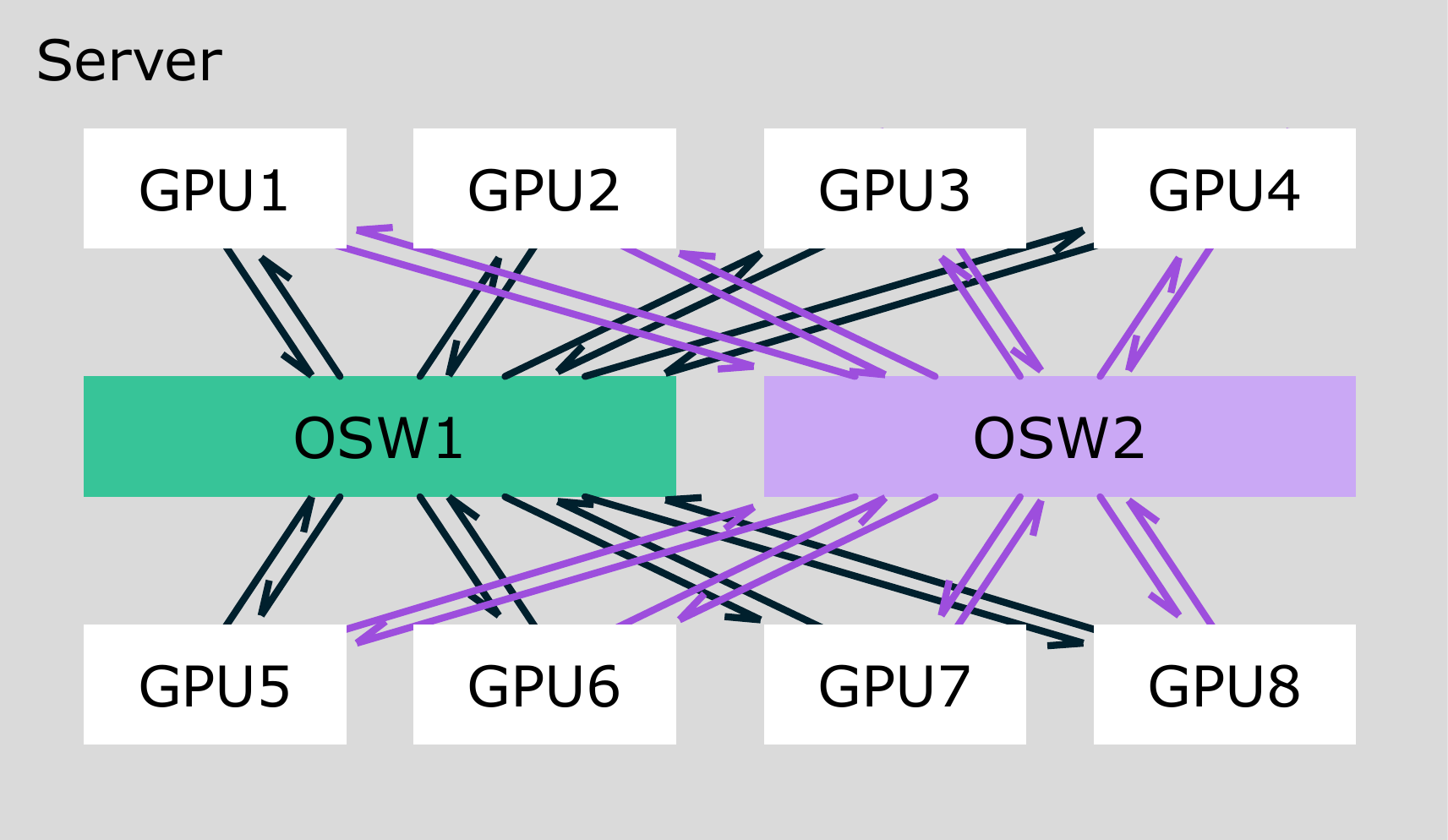}
\caption{An example scale-up network with $n=8$ GPUs interconnected by $k=2$ optical switches (OSW).}
\label{fig:motivation_setup}
\vspace{-0.15in}
\end{figure}

We consider a scale-up network with $n$ GPUs interconnected by $k$ optical switches, as shown in Figure~\ref{fig:motivation_setup}. Every GPU has both incoming and outgoing links to every switch. Therefore, at any time, the scale-up network realizes a directed graph of size $n$ and degree-$k$\footnote{Degree is defined as the number of outgoing links per GPU. We use "topology" and "graph" interchangeably.}.
The connections inside each optical switch (OSW) can be reconfigured over time, and each reconfiguration corresponds to a distinct size-$n$ degree-$k$ graph. Reconfiguration incurs a delay ranging from nanoseconds to microseconds~\cite{10.1145/3387514.3406221,10.1145/3651890.3672273,10.1145/3696348.3696856}, through milliseconds~\cite{polatis}, and up to several minutes~\cite{zhao2025direct-connect}, depending on hardware technology and port counts~\cite{liao2025mixnet,zhao2025direct-connect}. 
We target the all-to-all collective, where every GPU sends data to all other $n-1$ GPUs. We refer to each source-destination communicating pair as a \textbf{flow}. Each flow is divided into fixed-sized \textbf{chunks}, chunks traverse the network hop-by-hop using store-and-forward, and different flows may contain different numbers of chunks.

We define a \textbf{strategy} to be a sequence of topologies and the flow schedules that will run on these topologies,
where each topology is a size-$n$ degree-$k$ graph and each schedule consists of multiple rounds of flows. Figure \ref{fig:motivation_one} shows one example of a strategy, consisting of one topology and its associated schedule with seven rounds (ignore the bottom-row matrices for now). Each \textbf{round} routes multiple flows over one or multiple \textbf{hops}. It specifies, for each hop, which chunk is sent over which link.
Reconfiguration happens only after all flows in the current round complete. % \maria{or in the last round of the flow schedule associated with the topology}. 
For simplicity, we refer to the initial topology as the first reconfiguration. 
% We consider communication that proceeds in \emph{barrier-synchronized} manner: in the beginning of each round, all flows scheduled for that round are released at the same time, and the next round or reconfiguration begins only after all flows in the current round complete. \maria{we do not yet know that there are multiple rounds per topology, maybe move this to after you talk about the schedule?} For simplicity, we count the initial topology as the first reconfiguration.
% Each round  and specifies which flows are active and their routes in each round\maria{maybe an example would help here}; all flows across all rounds in all schedules sum up to the all-to-all collective.
% \maria{Alternative way of introducing round: In this work, we define a strategy as a sequence of rounds, where each round consists of a topology and a set of flows that are set to start on that topology at the same time. A topology is hence used in multiple rounds?}
We model completion time using the conventional $\alpha-\beta$ cost model~\cite{hockney1994alpha-beta-cost}: the cost of transmitting per-unit data over a link is $\alpha+\beta$, where $\alpha$ is the single-hop latency and $\beta$ is the inverse of link bandwidth. 
The \textbf{cost of a strategy} is the total collective completion time, which is the sum of reconfiguration time and data transmission time. The reconfiguration time is the per-reconfiguration overhead multiplied by the number of reconfigurations, and the transmission time is the sum across all rounds of the time required by the bottleneck flow to finish.
Our \emph{objective} is to find a strategy that minimizes the total cost i.e., the completion time of all-to-all, on a given network of $n$ GPUs, $k$ optical switches.
\vspace{-0.15in}

% \vspace{-0.2in}
\subsection{Limitations of Existing Work} \label{sec:motivation_limitation}
\vspace{-0.05in}

%\myitem{Existing work simplistically constrains the design space.}
The problem outlined above is essentially a joint optimization of topology and flow (communication) schedules for all-to-all. %  that no prior work solves fully. \vamsi{I think better to remove no prior work solves fully, just to be safe. The section title and rest of the text explains it anyway.}
One line of existing work compresses the joint optimization along one dimension and optimizes only one of the two, leading to suboptimal performance.
Collective algorithms (\eg recursive doubling~\cite{thakur2005recursivedoubling}, Bruck algorithm~\cite{bruck1994bruck} and SpreadOut~\cite{pjevsivac2007spreadout}) largely disregard the underlying topology, while communication libraries~\cite{nccl,rccl,deepep,cowan2023mscclang} are topology-aware when picking links but remain largely template- and heuristic-driven, and collective synthesis~\cite{shah2023taccl,liu2024teccl,cai2021sccl,cao2025syccl} takes the network topology as a fixed input.
MixNet~\cite{liao2025mixnet} greedily constructs a single topology from a given traffic matrix while general topology designs~\cite{besta2014slimfly,blach2024polarfly,kim2008dragonfly,singla2012jellyfish,valadarsky2016xpander} do not optimize for collective completion.
% \cite{basu2024efficienta2a} solves for communication schedules given a fixed topology and all-to-all communication.
\cite{bojja2016costly} acknowledges the difficulty of jointly optimizing topologies and flow schedules and reduces the problem to routing design when the sequence of topologies and flow assignments are given.
Another line of work sidesteps the decision of when and how to reconfigure by making unrealistic assumptions about the reconfiguration overhead. \cite{basu2024efficienta2a,zhao2025direct-connect} assume the per-reconfiguration overhead to be effectively infinite, so reconfiguring is never worthwhile, which reduces the problem to static topology design and standard routing design with multi-commodity flow. On the other end of the spectrum,  \cite{lei2025fast,kulkarni2017minbvn} assume negligible reconfiguration overhead (or a fully connected fabric), reducing the problem to a minimum Birkhoff-von Neumann (BvN) decomposition. 

Other existing work addresses related but orthogonal problems. 
SiP-ML~\cite{khani2021sipml} and TopoOpt~\cite{wang2023topoopt} co-optimize topology with machine learning workload placement and parallelization strategies, while TopoOpt~\cite{wang2023topoopt} and Swing~\cite{de2024swing} also leverage structures specific to all-reduce collective.
Hybrid data center designs~\cite{farrington2010helios,liu2015hybrid} fall back to packet switching for difficult traffic and often disallow multi-hop forwarding over optical circuits, while more conventional reconfigurable data center designs~\cite{porter2013mordia,saran2024semi,mellette2017rotornet,mellette2020opera,zerwas2023duo} grapple with a different set of challenges \eg handling traffic uncertainty, serving both bulky and latency-sensitive traffic, routing with only local information \etc.
% \vspace{-0.1in}

\begin{figure}[tp]
\centering
\includegraphics[width=0.85\columnwidth]{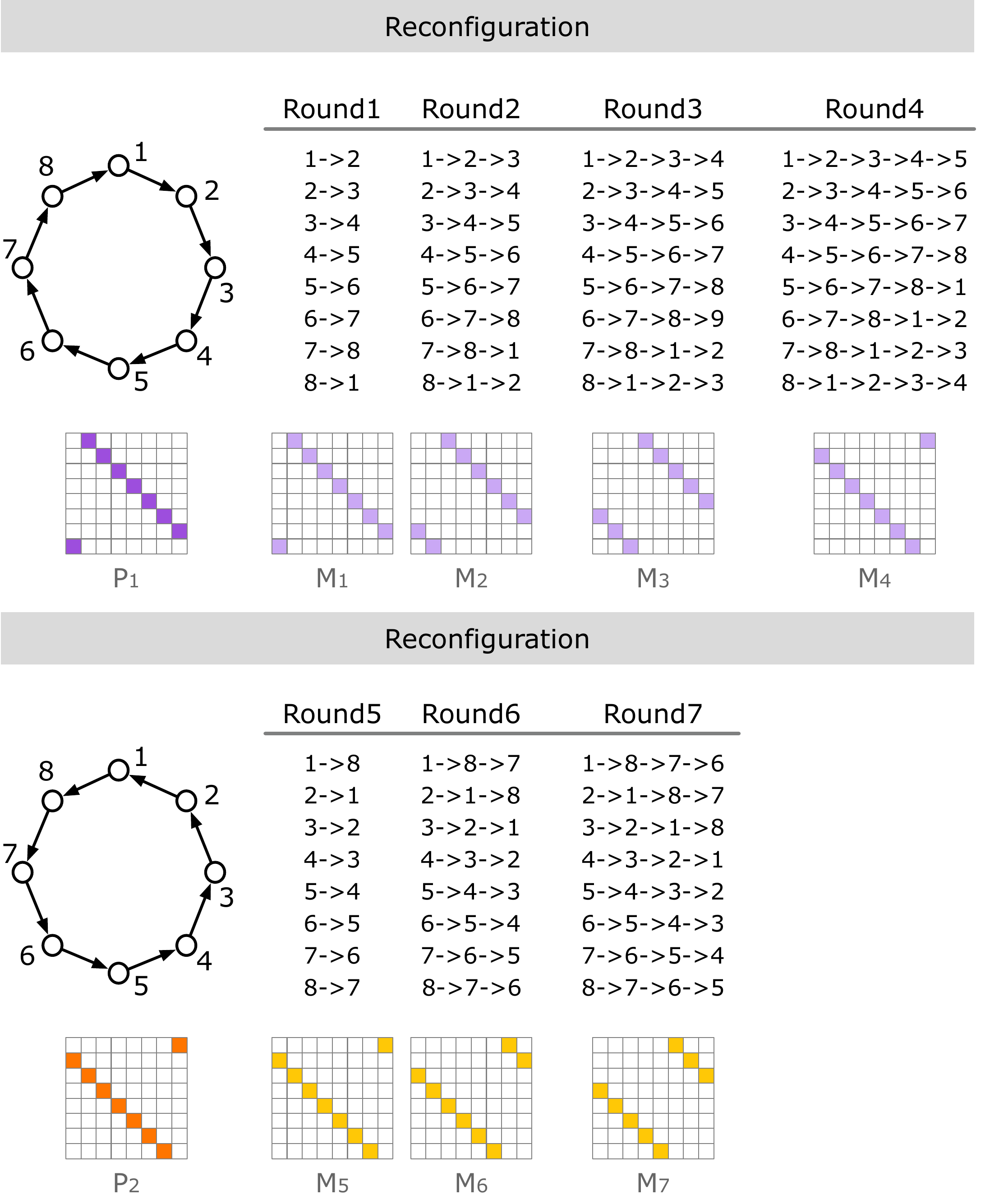}
\vspace{-0.2in}
\caption{An example strategy that reconfigures during the collective, which includes two topologies (reconfigurations) and schedules flows to run approximately evenly across the two topologies in specified rounds. Matrices below the topologies and rounds are illustrations of how they are captured (as adjacency matrices and their powers) in our abstraction.}
\label{fig:motivation_two}
\vspace{-0.2in}
\end{figure}

\begin{figure}[tp]
\centering
\includegraphics[width=0.98\columnwidth]{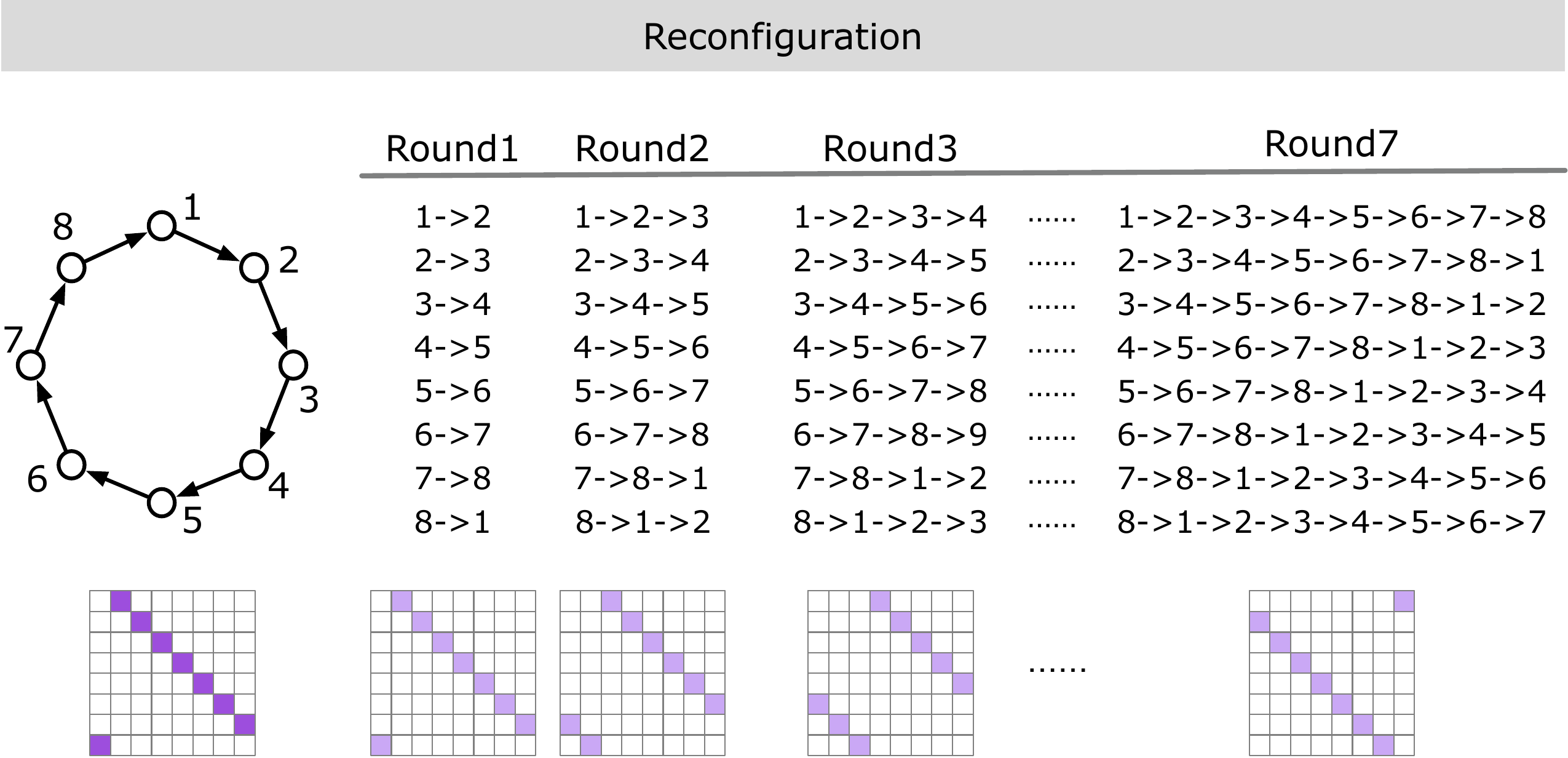}
\caption{An example strategy that reconfigures once at the beginning and sends all flows on this topology.}
\label{fig:motivation_one}
\end{figure}

\begin{figure}[tp]
\centering
\includegraphics[width=0.9\columnwidth]{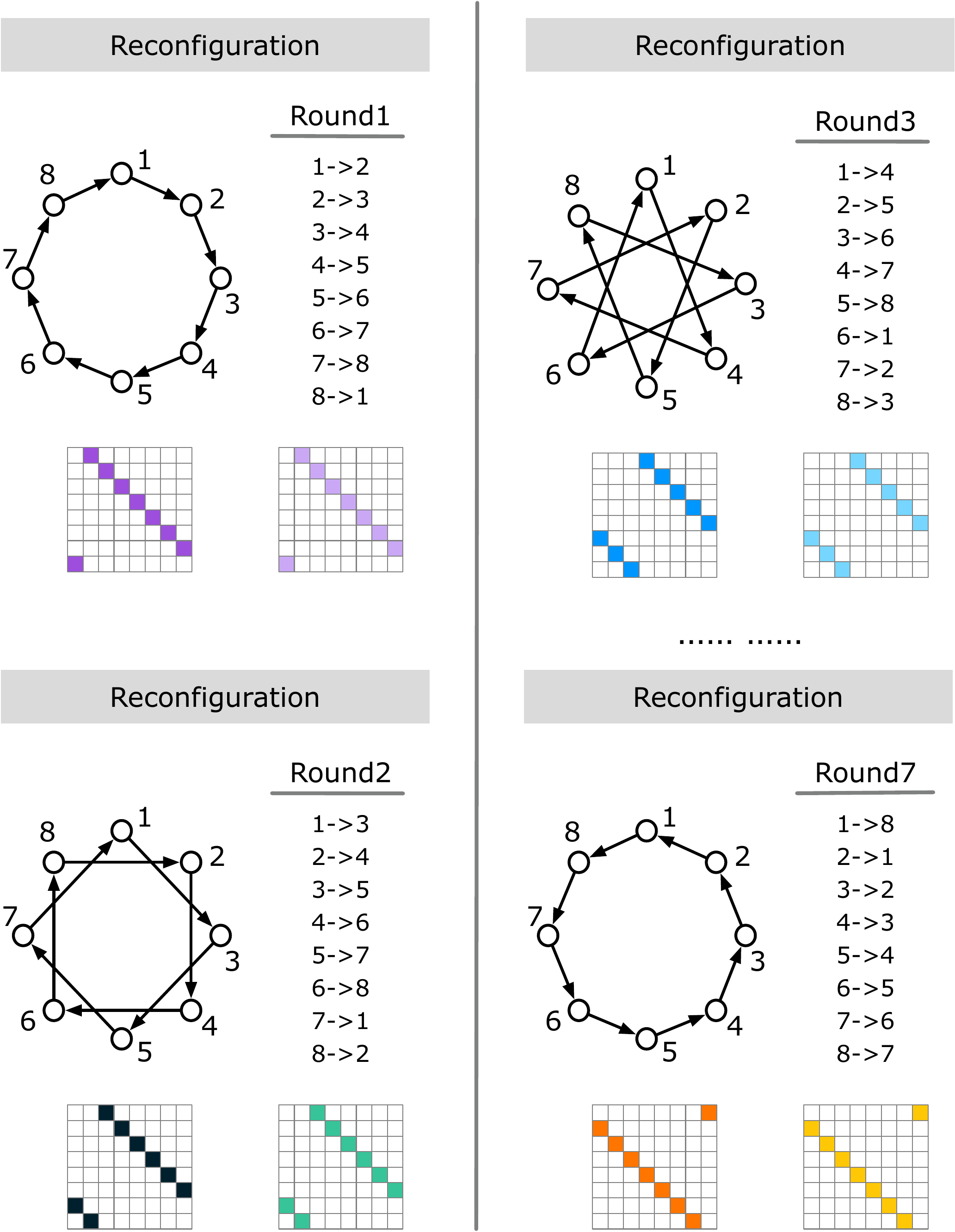}
\caption{An example strategy that reconfigures $n-1$ times and sends all flows on direct circuits. (Not all rounds are illustrated.)}
\label{fig:motivation_many}
\vspace{-0.2in}
\end{figure}

\vspace{-0.15in}
\subsection{Reconfiguring during Collective: Opportunity \& Challenge} \label{sec:motivation_example}
\vspace{-0.05in}
%\maria{chalge title to The case of reconfiguring during collective or something. My issue is that this contains not only the example but why is the problem hard}

%\myitem{Reconfiguring during the collective constitutes a major part of the design space.}

In this subsection, we use an example to demonstrate the potential benefit of reconfiguring during a collective and to show how prior work fails to exploit this opportunity due to simplifying assumptions. Crucially, we argue that these assumptions are not incidental or an oversight, but rather reflect the inherent difficulty of jointly optimizing topology and flow schedule.

%Prior work deals with the  reducing the problem to make it tractable, existing work misses a substantial opportunity: with realistic per-reconfiguration overheads, it could be worthwhile to reconfigure \emph{during} the collective, depending on whether the cost of transmitting flows over longer paths on the current topology exceeds the reconfiguration cost.

\myitem{Motivating Example:} Consider a scale-up network with 8 GPUs, 1 optical switch ($n=8$, $k=1$) where each GPU pair sends one chunk of data to each other, and $R=7T$, where $R$ denotes the per-reconfiguration delay and $T$ denotes the transmission time of sending one chunk over one hop.
%We will consider three approaches:\emph{(i)} allows reconfiguration during collective, \emph{(ii)} does not allow reconfigurations (equivalent to assuming reconfiguration costs is infinite), and allows any number of reconfigurations which are also illustrated in Fig. 

Figure~\ref{fig:motivation_two} illustrates a strategy that allows reconfiguration during collective (ignore the bottom-row matrices for now). This strategy schedules approximately half of the flows across four rounds on the initial topology, then reconfigures a new topology and schedules the rest of the flows across three rounds on that topology. On the first topology, each round costs $T, 2T, 3T, 4T$ respectively, and on the second topology, each round costs $T, 2T, 3T$ respectively, yielding a total cost of $16T + 2R=30T$.%\maria{can we say a bit more on how 30 is found \eg how much each round cost?} 

Now notice how assumptions \eg it is never/always worthwhile to reconfigure would inevitably lead to worse strategies. 
Indeed, under the assumption that it is never worthwhile to reconfigure, the best strategy is routing each flow on its shortest path and grouping flows with the same shortest path length in the same round as shown in Figure~\ref{fig:motivation_one}. That strategy will yield a total cost of $28T + R=35T$. 
On the other end of the spectrum, under the assumption that it is never worthwhile to reconfigure, the best strategy is to configure disjoint GPU pairings so that each GPU establishes a direct circuit to a distinct GPU in each round as shown in Figure~\ref{fig:motivation_many}. That strategy incurs a total cost of $7T + 7R=56T$. 
Overall, these strategies under simplistic assumptions yield at least $16\%$ and $86\%$ longer collective completion times, and this cost gap grows as $n$ increases.

%\myitem{Exhaustive search is prohibitively expensive.}
\myitem{The challenge of jointly optimizing topology and flow schedules.}
While comparing candidate strategies is straightforward, constructing them is not. Constructing good strategies requires deciding whether reconfiguration is beneficial, which topologies to deploy, which flows to run on the topologies, how often to reconfigure, \etc.
Even in our simple example, the optimal number of reconfigurations depends on how $R$ compares to $T$. More generally, the optimal sequence of topologies and flow schedules depends on the specific problem instance. As the problem instance grows larger and more complex, it becomes impossible to 
to infer by inspection a good strategy. Therefore, deciding when and how to reconfigure is inevitably a joint optimization over topologies and flow schedules.
Naively, we could enumerate over all possible (ordered) sequences of topologies and their associated schedules. However, even with a simplified setting that considers only how to partition flows across topologies without considering routing choices, exhaustive search is already infeasible. Specifically, there are at most $n!$ distinct degree-1 topologies (permutations). With $d$ reconfigurations, we have $n!(n!-1)(n!-2)...(n!-d+1) \approx n!^d (d \ll n!)$ possible length-$d$ topology sequences. Moreover, since all-to-all collective has no data dependency, there exists a staggering $d^{n(n-1)}$ possible ways to partition flows across $d$ topologies. The valid range of $d$ is from 1 to $n-1$. Overall, the search space size is 
\[
\sum_{d=1}^{n-1} (n!)^d \cdot d^{n(n-1)},
\]
which for $n=8$ is approximately $4.1\times 10^{79}$---far beyond what we could enumerate.
Related problems jointly optimizing topology and schedule are also known to be NP-hard \cite{foerster2018nphard,shah2023taccl}.

% The \uline{cost of a strategy} is the total collective completion time, which consists of reconfiguration time and data transmission time, written as
% \begin{equation} \label{eqn:totalcost}
% \textsf{TotalCost} = R\cdot  d + T\cdot \sum_{i=1}^d \sum_{j=1}^{s_i} f_{ij}\cdot C_{ij},
% \end{equation}
% where $R$ is the per-reconfiguration delay and $d$ is the number of reconfigurations; $T$ is the time to transmit per-unit \maria{one unit of?} data over one hop, $s_i$ is the number of rounds executed with topology $i$, $f_{ij}$ is the size of the bottleneck flow in round $j$ under topology $i$, and $C_{ij}$ is the corresponding multiplicative slowdown factor relative to transmitting that flow alone over one hop.
% $R$ and $T$ are network-specific, and flow sizes $f$ depend on the specific all-to-all collective; different strategies incur different $d$, and different $\sum_i \sum_j f_{ij}\cdot C_{ij}$ depending on how flows are partitioned and routed.
% Our \emph{objective} is to find a strategy that minimizes \textsf{TotalCost} on a given network of $n$ GPUs, $k$ OCSs, per-reconfiguration delay $R$, per-unit transmission delay $T$ with a given all-to-all collective.
%!TEX ROOT=../paper.tex
\section{Overview}

Instead of enumerating over individual strategies, we introduce an abstraction that models reconfiguring during a collective as a decomposition of the traffic demand into powers of adjacency matrices, enabling reasoning on families of strategies without enumeration (\S\ref{sec:sec3_abstraction}). %captures the entire strategy space, enabling reasoning on regions of strategies (\S\ref{sec:sec3_abstraction}). 
Next, we identify a special region in the abstract solution space characterized by highly symmetrical and expandable topology sequences, and explain our insights on efficiently constructing low-cost strategies (\S\ref{sec:sec3_strategies}).
To build intuition, we will focus on the setting of $k=1$ and equal-sized flows in this section, and discuss generalizations to arbitrary settings at the end of \S\ref{sec:sec3_strategies}.
\vspace{-0.1in}

\subsection{An Abstraction for Reconfiguring during Collective} \label{sec:sec3_abstraction} %\ann{not only during collective? technically the abstraction also includes one-shot reconfiguration}
%Allowing reconfiguration during collective communication explodes the potential strategy space. 
Instead of reasoning over the whole space of all concrete strategies and their cost, we introduce one level of abstraction.
%Therefore, a promising path forward---and the one taken in this paper---is to come up with an abstract way of representing the entire strategy space.
% Therefore, the core challenge lies in coming up with an abstract way of representing the entire strategy space. \maria{Is this a requirement or your contribution? The sentence before says it is a requirement, but the one bellow claims you use it for something-- not all abstractions can be used like that}
%Such an abstraction \maria{not all abstractions do that} 
Our abstraction excludes invalid strategies (\eg strategies that do not respect $k$, do not complete all flows in the collective, \etc), and does not keep track of lower-level routing details. 
Consequently, it exposes the underlying structure of the problem, and enables us to reason about groups (or regions) of strategies that share common topology sequence and flow assignment to topologies.
Critically, our abstraction yields an explicit formulation of the total cost and enables us to derive a lower bound on total cost globally. %Next, we derive a representation for the strategy space, which leads us to an explicit formulation for total cost and a lower bound on total cost.
%\maria{I think last 2 sentences are repetitive}

% \maria{You need to say here why do you go back to the example: Maybe you will first show how your representation allows you to calculate the cost of the strategy of the example without enumerating each round and its routing? Or start with the formula and then apply to the example}
% \myitem{Matrix-based representation.}
We will start with describing the strategy in Figure~\ref{fig:motivation_two} with matrices, show that the strategy corresponds to a decomposition of the all-to-all traffic matrix, and that the cost of the strategy can be derived from the decomposition. Finally, we derive a general matrix representation for the strategy space, which leads to an explicit formulation for and a lower bound on total cost.
Let a matrix $A\in\mathbb{R}_{\ge 0}^{n\times n}$ represent the flows in an all-to-all collective.
A \emph{strategy (without routing details)} can be represented as a sequence $S = \langle (M_1, P_1), (M_2, P_2), ... \rangle$ where $M_i$ is the flow matrix in round $i$ and $P_i$ is the topology used in round $j$. %\ann{but strategy also has a routing component}\maria{you have defined strategy before maybe you could include M and P there} 
Any valid strategy corresponds to a matrix decomposition $A = \sum_i M_i$.
For example, the strategy in Figure \ref{fig:motivation_two} corresponds to 
\vspace{-0.05in}
\[
A = M_1 + M_2 + M_3 + M_4 + M_5 + M_6 + M_7.
\]
\vspace{-0.05in}
Each $(M_i, P_i)$ of round $i$ has an associated transmission cost $\geq h$, where $M_i=P_i^h$ indicating $h$ hop for $M_i$ to be served on $P_i$. 
For instance, in Figure \ref{fig:motivation_two}, each flow in round 1 traverses one distinct hop on topology $P_1$, and thus we write $M_1 = P_1$ and round 1 costs $T$ ($T$ denotes the transmission time of sending one chunk over one hop). Similarly, each flow in round 2 traverses two hops on $P_2$ and we write $M_2 = P_1^2$; round 2 incurs $2T$ transmission cost as hops do not overlap in the time-expanded graph. In general, multi-hop forwarding can be represented by taking powers of the underlying topology, allowing a further simplification of the decomposition to
\[
A = P_1 + P_1^2 + P_1^3 + P_1^4 + P_2 + P_2^2 + P_2^3.
\]
This decomposition maps to a cost calculation that takes the sum of the matrix powers as the transmission cost and the number of distinct $P_i$ as the reconfiguration cost, yielding a total cost of $2 R + 16 T$.
Likewise, the other two strategies in our previous example also correspond to a matrix decomposition and a corresponding total cost, \ie
\[
\text{Fig.\ref{fig:motivation_one}}: A = P_1 + P_1^2 + P_1^3 + P_1^4 + P_1^5 + P_1^6 + P_1^7, 
\quad
\textsf{TotalCost} = R + 28 T.
\]
\[
\text{Fig.\ref{fig:motivation_many}}: A = P_1 + P_2 + P_3 + P_4 + P_5 + P_6 + P_7, 
\quad
\textsf{TotalCost} = 7 R + 7 T.
\]

More generally, we claim that any strategy corresponds to a decomposition of $A$: 
\[
A = \sum_{i=1}^{d} \sum_{j=1}^{s_i} \sum_{l=1}^{g_{ij}} M_{ijl} \odot P_i^{h_{ijl}},
\]
where $A\in\mathbb{R}_{\ge 0}^{n\times n}$ denotes the all-to-all traffic matrix with zero diagonal, $P_i$ denotes $i$-th topology, a non-zero entry in $P_i^h$ denotes a flow taking an $h$-hop route on $P_i$. 
$g_{ij}$ is the number of groups of flows by hop count transmitted in round $j$ on topology $i$, $h_{ijl} \geq 1$ is the hop count of group $l$ in round $j$ on topology $i$, $M_{ijl}\in\mathbb{R}_{\ge 0}^{n\times n}$ is a mask that selects which flows are sent in that group. Flows from different groups in the same round are started at the same time. $\odot$ is the element-wise product. 
Intuitively, each term $M_{ijl} \odot P_i^{h_{ijl}}$ represents a group of flows that are routed using $h_{ijl}$ hops over topology $P_i$ in round $j$.
This matrix-based representation abstracts away lower-level routing details of a strategy, and consequently, given any strategy $S_d$, we can derive its matrix composition and express its total cost as  
\[
\textsf{TotalCost}(S_d) \geq R\cdot d + T\cdot \sum_{i=1}^{d} \sum_{j=1}^{s_i} \max_l h_{ijl}.
\]
When flows in the same round do not share the same link at the same hop, there is no bandwidth \emph{contention} in the time-expanded graph, and we obtain an exact equation for \textsf{TotalCost}.

% \myitem{Lower bound on global optimum.}
In addition, based on this explicit formulation of \textsf{TotalCost}, Theorem~\ref{thm:lowerbound} derives a lower bound on the minimum total cost across any valid strategy given $n$, $k$ and $d$ \ie the global optimum of the strategy space.

% \begin{theorem}[Lower bound on \textsf{TotalCost}$(n,k,d)$] \label{thm:lowerbound}
% Given the number of GPUs $n$, number of optical switches $k$, number of reconfigurations $d$, per-reconfiguration time $R$, and the transmission time of one unit of data over one link $T$, the total cost of any strategy cannot be lower than the following:
% \begin{equation*} \label{eqn:mintotalcost}
% \textsf{TotalCost}(n,k,d) \geq R\cdot d + T\cdot\big(
%  d\cdot\frac{q(q+1)}{2} + u\cdot(q+1)\big),
% \end{equation*}
% where \(q=\left\lfloor\frac{n-1}{d}\right\rfloor, u=(n-1)\bmod d\).
% \end{theorem}

\begin{theorem}[Lower bound on \textsf{TotalCost}(n,k,d)]\label{thm:lowerbound}
Fix integers \(n\ge 2\), \(k = 1\), and \(d\ge 1\). Let \(R>0\) denote the per-reconfiguration time and \(T>0\) denote the transmission time per unit data per hop. 
For any valid strategy that performs exactly \(d\) reconfigurations, the total completion time satisfies
\[
\textsf{TotalCost}(n,k,d)\;\ge\; dR \;+\; T\!\left( d\cdot \frac{q(q+1)}{2} \;+\; u\cdot (q+1)\right),
\]
where \(q=\left\lfloor\frac{n-1}{d}\right\rfloor\) and \(u=(n-1)\bmod d\).
\end{theorem}

\begin{proof}[Proof Sketch.]
\textsf{TotalCost} is minimized when $\sum_i \sum_j \max_l h_{ijl}$ is minimized and every round is \emph{contention}-free. $\sum_i \sum_j \max_l h_{ijl}$ is minimized (in a hypothetical idealized setting) when both of the following hold.
% (1) we simply use the shortest paths from a single reference source to all its destinations, and assume that when all sources act like that reference source, the resulting schedule achieves the same completion time, 
(1) Each flow is routed along its shortest path, and all sources in the topology are symmetric and see an isomorphic connectivity pattern; hence, the flow completion time is identical for all sources and can be characterized by a single representative source.
(2) The $d$ topologies have maximal reachability, so that these shortest paths that flows traverse on are indeed globally shortest among all possible topologies.
Each source has $k=1$ outgoing link, so delivering to all $n-1$ destinations requires at least $n-1$ rounds in total (across all $d$ reconfigurations).
On each topology, we should leverage paths of the smallest hop counts without repeating (as repeating hop counts lead to repeating destinations). Therefore, to minimize the total hop counts over all rounds, the $n-1$ rounds should be partitioned across the $d$ reconfigurations as evenly as possible.
\end{proof}
% \vspace{-0.15in}

\subsection{Constructing Low-Cost Strategies} \label{sec:sec3_strategies}
Unfortunately, a decomposition corresponding to the lower bound \textsf{TotalCost}$(n,k,d)$ does not always exist. For example, when $n=8, d=3$, there exist no matrices $P_1,P_2,P_3$ such that $A=P_1+P_1^2+P_1^3+P_2+P_2^2+P_3+P_3^2$, meaning no strategy can achieve the minimum possible \textsf{TotalCost} here. This naturally raises the question of how to construct strategies that operate at or close to the global optimum, given network parameters $n,k$ and the number of reconfigurations $d$.
Fortunately, our abstraction still guides us towards regions of the strategy space that yield low-cost strategies and can be identified with low computation overhead. 
We will again refer to Figure \ref{fig:motivation_two} for intuition. Observe that with topologies $P_1$ and $P_2$, if each flow picks the topology that provides shorter shortest paths, and flows with the same shortest path lengths are grouped into the same round, the resulting schedule is globally optimal, as all flows in this schedule traverse the shortest paths possible and all links are fully utilized.
This observation leads to our key insight that \emph{an appropriate sequence of topologies} allows an efficient algorithm to identify a schedule that achieves the minimum possible total cost on this topology sequence. 
Critically, this regional optimum is also close to the global optimum. 
Consequently, we reduce the original two-dimensional search over both topologies and schedules to the simpler task of identifying good topology sequences. 
Next, we will elaborate on our design principles on how to construct such sequences---ultimately allowing us to avoid explicit search over topologies altogether---and how to derive the corresponding schedule.

\myitem{Design principle \#1: Intra- \& inter-topology symmetry.}
When nodes in a topology have isomorphic connectivity pattern (intra-topology symmetry), their available paths are similar. As a result, available path lengths are more uniform across flows originating from different sources. This prevents a small number of outstanding flows that require longer paths from dominating the schedule's completion time by prolonging rounds in which most flows complete earlier.
Moreover, different topologies in the sequence may provide shorter paths for different source-destination pairs. When these topologies are symmetrical to each other (inter-topology symmetry), these shorter path opportunities eliminate entire rounds rather than benefiting only a subset of flows within a round, leaving the round's completion time unchanged. 
Together, these intra- and inter-topology symmetries reduce the number of rounds in a schedule by ensuring that rounds are densely populated \ie topology capacity is efficiently utilized with minimal waste, thereby reducing total cost.

\myitem{Design principle \#2: High expansion.}
Each topology should exhibit strong expansion \ie each node should reach as many distinct destinations as possible within a given hop-count limit, thus supporting short paths for any source-destination pair. (This is trivial for the degree-1 case since it can only be a ring, but crucial for the general degree-$k$ graphs.)
Moreover, the topology sequence should exhibit strong expansion \ie different topologies in the sequence should expose largely non-overlapping $h$-hop neighborhoods for each given source, so that reconfigurations meaningfully expand the destinations reachable within $h$ hops, rather than repeatedly covering the same nodes. Consequently, the shortest paths available across topologies in the sequence are indeed short.
High expansion reduces the path lengths required in each round and thus reduces total cost.

\myitem{Design principle \#3: Contention-free schedule.}
A good schedule should increase link utilization across the network while avoiding contention (in the time-expanded graph). To achieve this, we develop an efficient algorithm that assigns flows to their shortest paths, groups flows by path length, and schedules contending paths to different rounds. 
For highly expandable graphs (that we use in the degree-$k$ case) that inevitably lack perfect intra-topology symmetry, our algorithm opportunistically combines rounds that do not contend with each other, reducing the number of rounds further. 
Intra- and inter-topology symmetries ensure that this scheduling does not incur an excessive number of rounds, while high expansion keeps the hop counts within each round small.

\myitem{Generalization to degree-$k$ graphs \& varied-sized flows.}
% For degree-$k$ graphs, symmetry and expansion could be hard to co-exist in a topology, and thus we make use of two families of topologies---union of shifted rings (stronger symmetry) and generalized Kautz graphs (stronger expansion)---depending on the amount of data to transmit on a topology. \maria{weird that you say it is hard to co-exist do you mean in a random topology? Maybe you want to say this type of graph achives a good trade-off on the properties we aim to preserve \ie symetry and expansion?}
For degree-$k$ graphs, we leverage two types of topologies: \emph{(i)} circulant graphs, which exhibit strong symmetry but lack high expansion within each graph due to their regular structure, and \emph{(ii)} generalized Kautz graphs, which exhibit high expansion but lack perfect symmetry. We construct the topology sequence based on circulant when traffic is lighter and rounds are reasonably well-packed, and switch to construction based on generalized Kautz when traffic is heavier hence shorter paths are more crucial.
% As a result,  more care is taken when constructing the topology sequence, and when deriving the schedule.
For all-to-all with varied-sized flows, we optimize the schedule by grouping larger flows to the same round, essentially relabeling the nodes in a topology.

% \myitem{Accounting for the number of reconfigurations $d$.}
% Ultimately, which strategy is best depends on the ratio between per-reconfiguration cost $R$ and data transmission cost $T\cdot f$, which varies across networks and collectives and may not be known until runtime.
% Therefore, we divide our approach into two phases. In the first \emph{offline} precomputation phase, we minimize the total cost for any given number of reconfigurations $d$, producing a family of candidate strategies parameterized by $d$. In the second \emph{online} selection phase, we choose the best $d$ and the corresponding strategy when network-specific $R, T$ and collective-specific $f$ are known. \ann{check whether holds for arbitrary}
\vspace{-0.1in}
\section{Design}
\vspace{-0.05in}

This section starts with a design overview (\S\ref{sec:design_overview}) and provides details on how to select an appropriate type of topology (\S\ref{sec:design_topochoice}) and how to construct the topology sequence (\S\ref{sec:design_toposeq}).
We will first present the degree-1 case to build intuition, and then extend the discussion to the degree-$k$case. 
While we describe them separately for clarity, the two settings follow the same underlying principles.
Our schedule generation is tightly coupled with the topology generation, and its details are embedded in \S\ref{sec:design_topochoice} and \S\ref{sec:design_toposeq}. 
We close the section by discussing how to optimize for varied-sized flows (\S\ref{sec:design_opt}).
% Degree-1 is a simplifed version of degree-$k$. The choice of topologies at degree-1 and at degree-larger-than-1 slightly differs, with degree-larger-than-1 being the most challenging. 
% So we discuss them both. 
% In each subsection, we discuss degree-1 first for intuition and then generalize to degree-$k$. We first tackle uniform all-to-all collective and extend our approach to arbitrary all-to-all 
% as an optimization

\vspace{-0.1in}

\subsection{Design Overview} \label{sec:design_overview}

The total cost of a strategy is parameterized by the number of reconfigurations $d$, as shown in Theorem~\ref{thm:lowerbound}. Therefore, we generate a strategy for each possible $d$ values, compute its total cost, and pick the best $d$ (and the corresponding strategy) that yields the lowest cost.
We generate the topology sequence in two steps. First, we select a base topology that exhibits strong intra-topology symmetry and expansion. Second, we generate a sequence of $d-1$ topologies by "shifting" or "contracting" based on the base topology, in a way that ensures strong inter-topology symmetry and expansion.
We use \emph{contention} to refer to link sharing in the time-expanded graph. Both base and shifted topologies admit bandwidth-efficient, contention-free schedules for all-to-all.
Our schedule generation is more efficient than state-of-the-art schedulers that either use constraint-based solvers (\eg TACCL~\cite{shah2023taccl}, TECCL~\cite{liu2024teccl}) or require solving an NP-hard minimum BvN decomposition problem (FAST~\cite{lei2025fast}).
\vspace{-0.1in}

\begin{figure}[t]
     \centering
     \begin{subfigure}[b]{0.35\columnwidth}
         \centering
         \includegraphics[width=\textwidth]{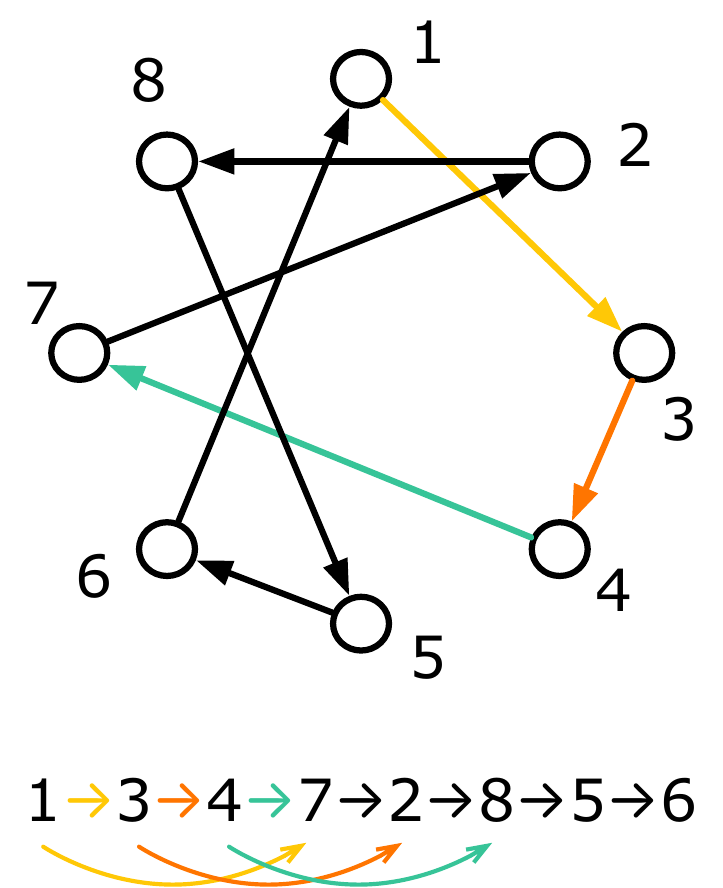}
         \caption{Base topology.}
         \label{fig:design_deg1topobase}
     \end{subfigure}
     % \hfill
     \begin{subfigure}[b]{0.35\columnwidth}
         \centering
         \includegraphics[width=\textwidth]{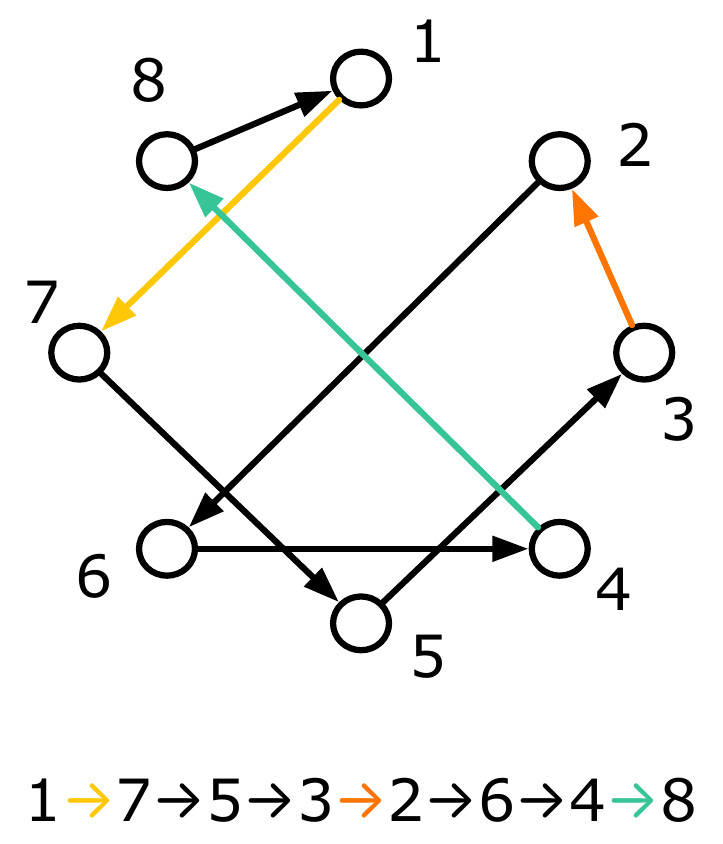}
         \caption{Shifted topology.}
         \label{fig:design_deg1toposhift}
     \end{subfigure}
     \vspace{-0.1in}
     \caption{An example of base and shifted topologies for degree-1. }
     % \vspace{-0.15in}
\label{fig:design_deg1topo}
\vspace{-0.3in}
\end{figure}

\begin{figure}[tp]
  \centering

  \begin{subfigure}{\linewidth}
    \centering
    \includegraphics[width=0.98\linewidth]{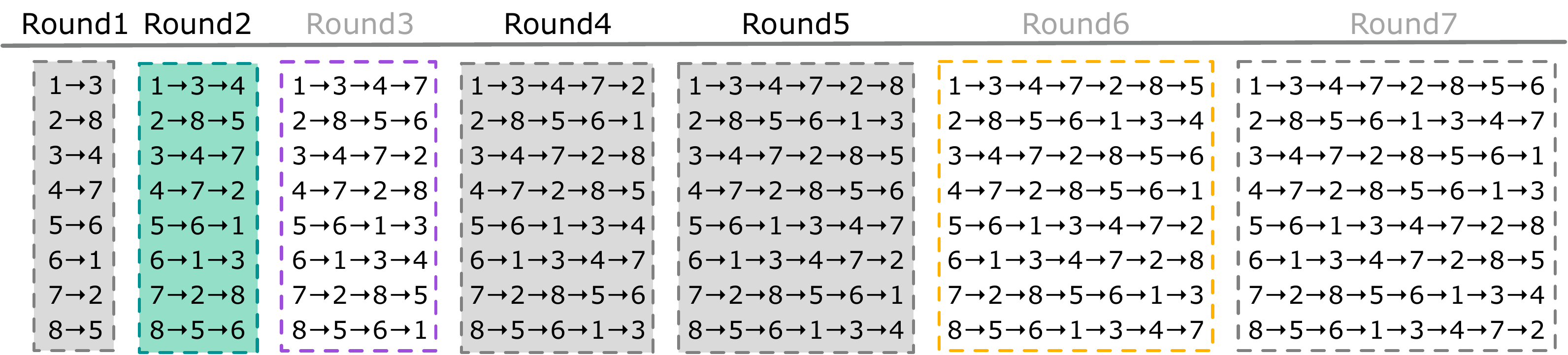}
    \caption{Schedule for the base topology.}
    \label{fig:design_deg1schbase}
  \end{subfigure}

  % \vspace{} % adjust vertical gap

  \begin{subfigure}{\linewidth}
    \centering
    \includegraphics[width=0.98\linewidth]{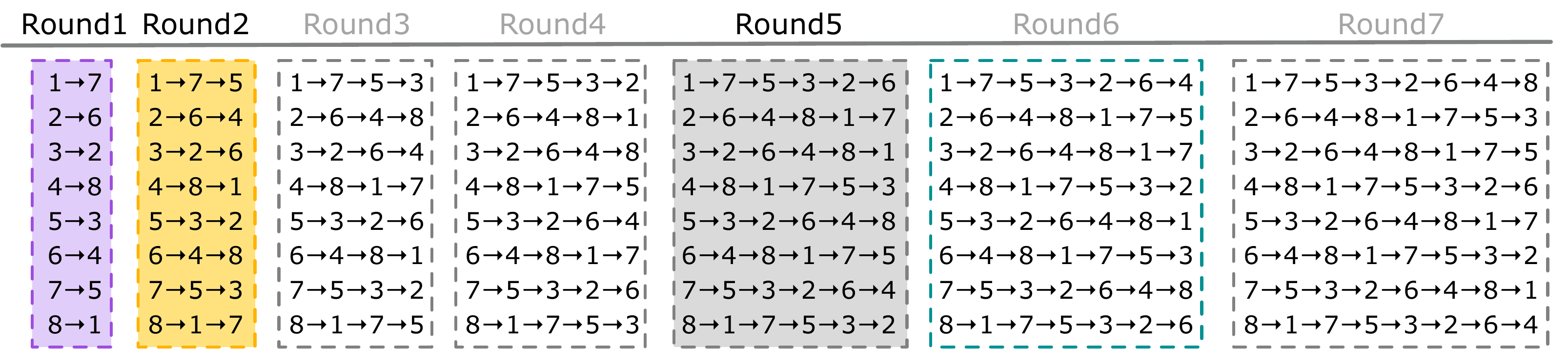}
    \caption{Schedule for the shifted topology.}
    \label{fig:design_deg1schshift}
  \end{subfigure}

    \vspace{-0.1in}
  \caption{Contention-free, bandwidth-efficient schedules that correspond to the two topologies in Figure~\ref{fig:design_deg1topo}.}
  \label{fig:design_deg1sch}
  \vspace{-0.2in}
\end{figure}

\subsection{Choice of Base Topology \& Scheduling on a Topology} \label{sec:design_topochoice}

\myitem{Degree-1.} 
Degree-1 graph is essentially a permutation over $n$, hence achieving intra-topology symmetry is trivial.
To ensure high expansion, we pick a permutation that forms a single $n$-cycle, % instead of say two $\frac{n}{2}$-cycles.
which is denoted by $\langle \sigma_1, \sigma_2, \sigma_3, ... \sigma_n\rangle$ i.e., node $\sigma_1$ is connected to node $\sigma_2$, node $\sigma_2$ is connected to node $\sigma_3$, and so on. %Note that if $\sigma_1=\sigma_2=\sigma_3 = ... = \sigma_n$, then this is just a uniformly shifted ring. But we don't necessarily need that. 
Figure~\ref{fig:design_deg1topobase} shows an example of the base topology with its corresponding $\sigma$-sequence (ignore the colors for now).

All flows can only traverse one path in this case, and we generate the schedule by grouping all flows with the same path length into the same round. Due to intra-topology symmetry \ie nodes have isomorphic connectivity patterns, the paths from a source to different destinations are equivalent to those from other nodes. As a result, these rounds are fully packed and contention-free. Figure~\ref{fig:design_deg1schbase} shows the schedule for the example base topology (ignore the colors for now). Critically, this schedule achieves the minimum possible transmission time, \ie is globally optimal.

\begin{figure}[t]
     \centering
     \begin{subfigure}[b]{0.35\columnwidth}
         \centering
         \includegraphics[width=\textwidth]{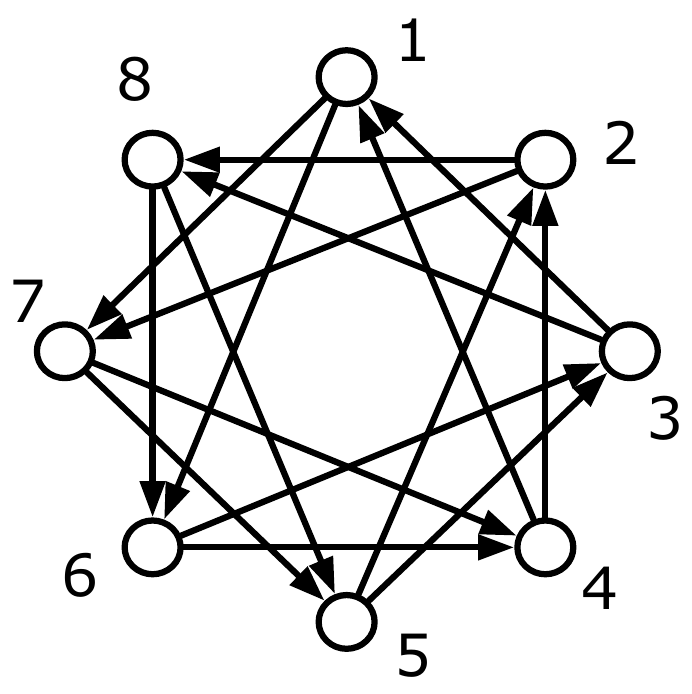}
         \caption{Base Circulant topology.}
         \label{fig:design_degkcirc_topobase}
     \end{subfigure}
     % \hfill
     \begin{subfigure}[b]{0.35\columnwidth}
         \centering
         \includegraphics[width=\textwidth]{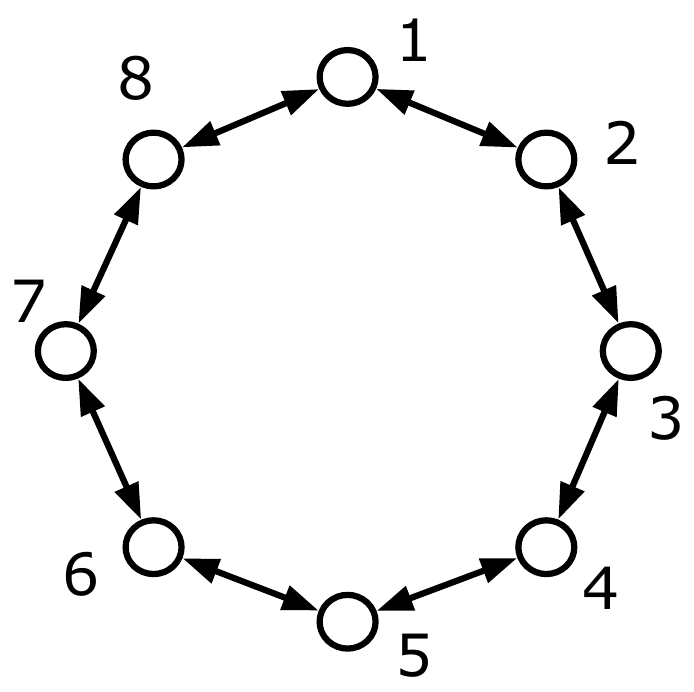}
         \caption{Shifted Circulant topology.}
         \label{fig:design_degkcirc_toposhift}
     \end{subfigure}

     % \vspace{0.6em} % adjust vertical gap

  \begin{subfigure}{\linewidth}
    \centering
    \includegraphics[width=0.5\linewidth]{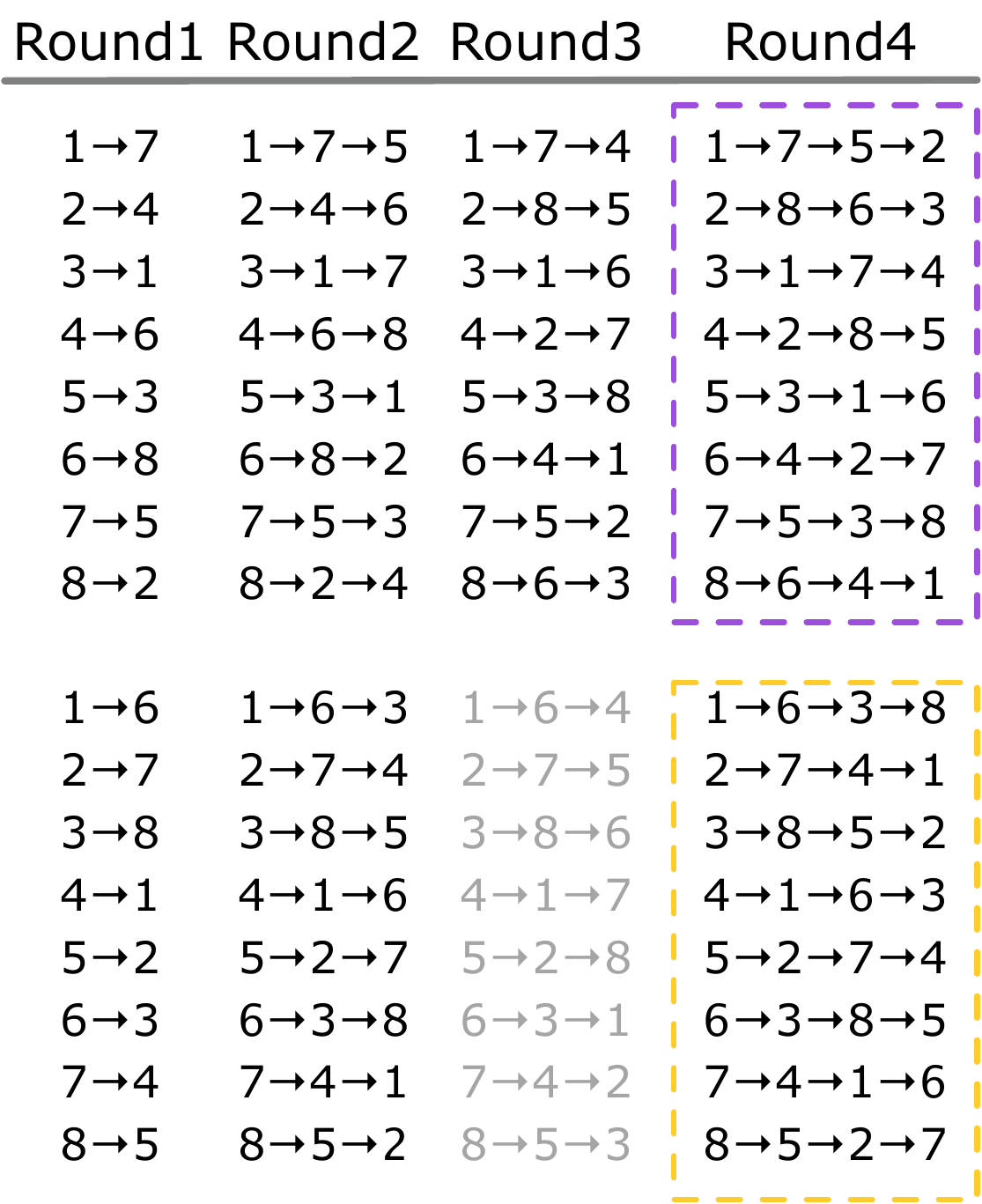}
    \caption{Schedule for the base Circulant topology.}
    \label{fig:design_degkcirc_schbase}
  \end{subfigure}
  
     \vspace{-0.1in}
     \caption{An example of base and shifted Circulant topologies, and the schedule for the base topology, for degree-$k=2$.}
     % \vspace{-0.15in}
\label{fig:design_degkcirc}
\vspace{-0.2in}
\end{figure}

\begin{figure}[t]
     \centering
     \begin{subfigure}[b]{0.35\columnwidth}
         \centering
         \includegraphics[width=\textwidth]{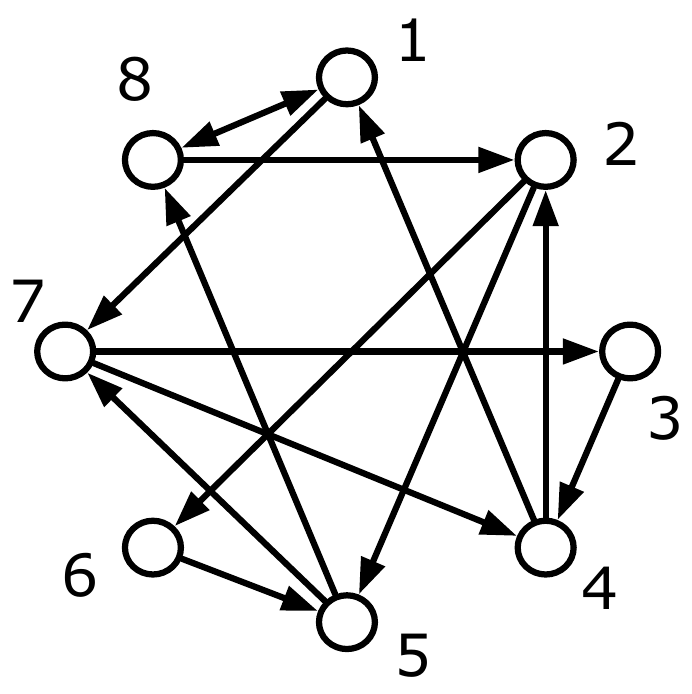}
         \caption{Base GenKautz topology.}
         \label{fig:design_degkgkau_topobase}
     \end{subfigure}
     % \hfill
     \begin{subfigure}[b]{0.35\columnwidth}
         \centering
         \includegraphics[width=\textwidth]{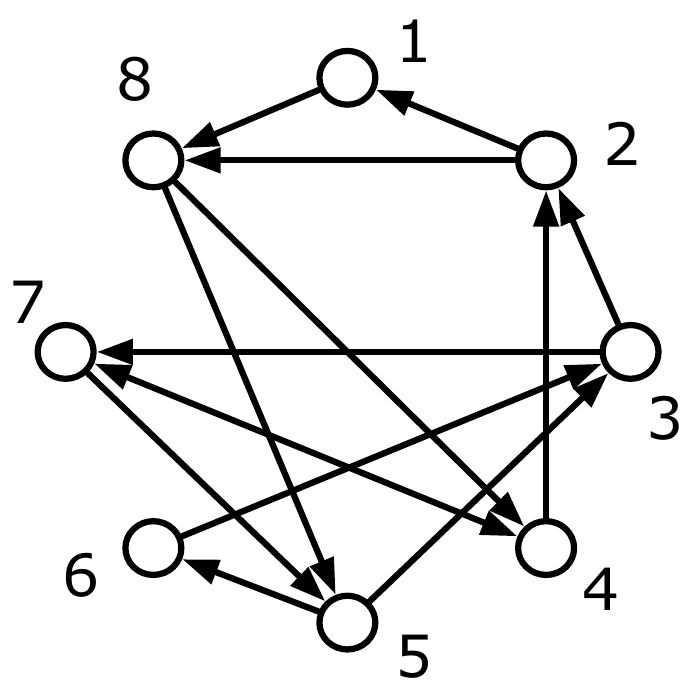}
         \caption{Shifted GenKautz topology.}
         \label{fig:design_degkgkau_toposhift}
     \end{subfigure}

     % \vspace{0.6em} % adjust vertical gap

  \begin{subfigure}{\linewidth}
    \centering
    \includegraphics[width=0.98\linewidth]{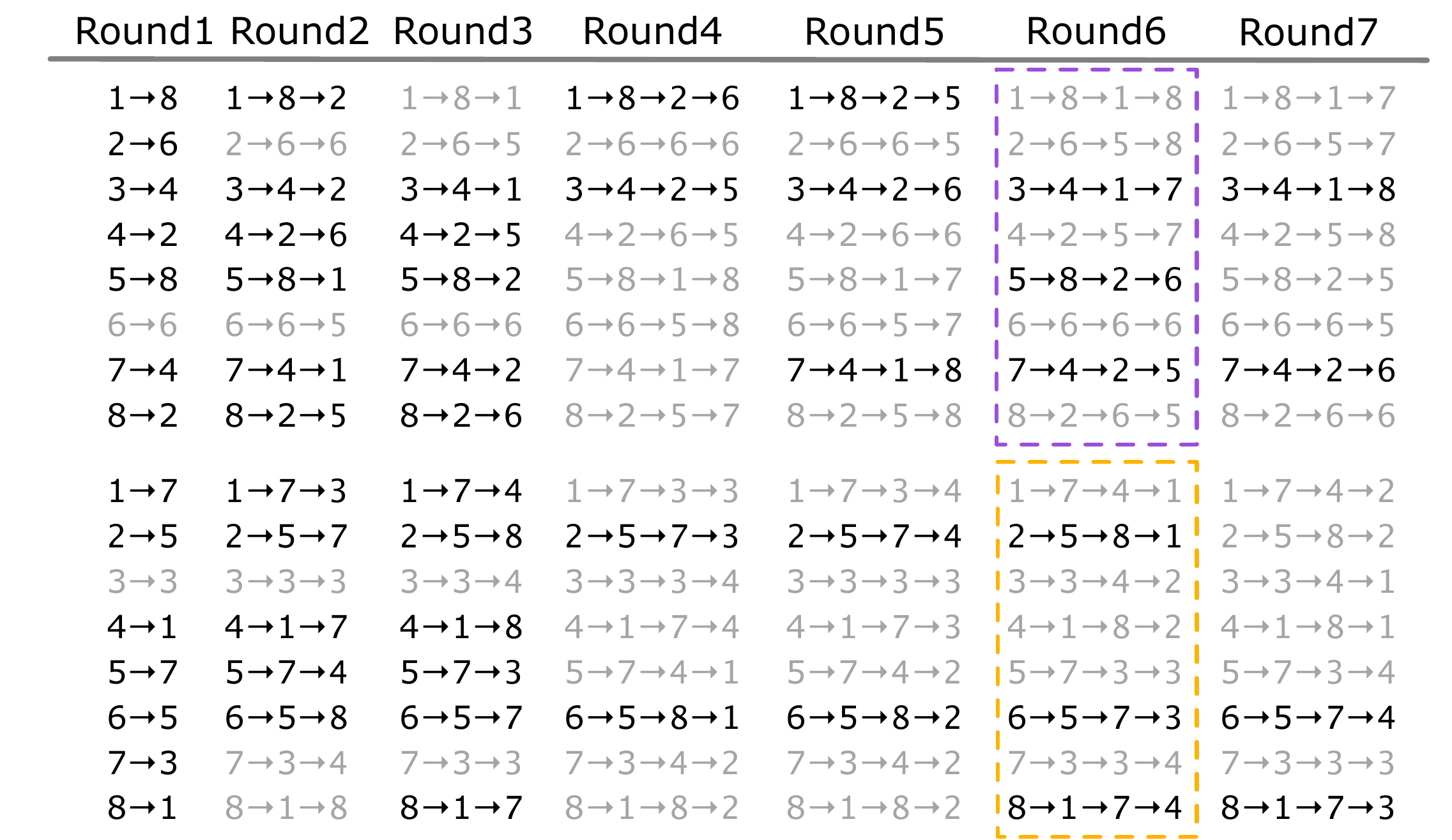}
    \caption{Schedule for the base GenKautz topology.}
    \label{fig:design_degkgkau_schbase}
  \end{subfigure}
  
     \vspace{-0.1in}
     \caption{An example of base and shifted Generalized Kautz topologies, and the schedule for the base topology, for degree-$k=2$.}
     % \vspace{-0.15in}
\label{fig:}
\vspace{-0.2in}
\end{figure}

\myitem{Degree-k.} 
We use two types of graphs for degree-$k$: Circulant graphs and Generalized Kautz graphs.
Figure~\ref{fig:design_degkcirc_topobase} shows an example of the base Circulant topology constructed for $k=2$. Circulant graphs exhibit perfect intra-node symmetry, and the paths from a source to different destinations are equivalent to those from other nodes too. However, this regular structure inevitably leads to limited expansion per topology.
When generating the schedule, we consider node1 as a reference node, and list the destinations that node1 can reach when taking a 1-hop path, 2-hop path, \etc, until all destinations have been covered. When multiple paths lead to the same destination, we pick the one of the shortest length, and break the tie randomly. After obtaining $n-1$ paths to the $n-1$ destinations, we sort these paths by their path lengths and start packing them into rounds. Each round has a capacity of $k$. We greedily pack the $n-1$ paths into as few rounds as possible, creating a new round only when previous rounds are either full or contending. Since a Circulant graph is symmetrical, we populate the schedule for other nodes by following the paths chosen by the reference node. Figure~\ref{fig:design_degkcirc_schbase} shows the schedule for the base Circulant topology (ignore the colored boxes). One group of paths is colored gray as they are not selected by the schedule. Notice how symmetrical graphs like Circulant ensure that each round is as packed as possible.

Figure~\ref{fig:design_degkgkau_topobase} shows an example of the base Generalized Kautz topology constructed for $k=2$. Generalized Kautz (GenKautz) graph \footnote{GenKautz is isomorphic to de Bruijn graphs when $n = k^D$, where $D$ is the diameter, but allows more flexible pairs of $n$ and $k$.} is an expander, thus ensuring high expansion per topology, but inevitably exhibits weaker intra-node symmetry. 
Figure~\ref{fig:design_degkgkau_schbase} shows the schedule for the base GenKautz topology (ignore the colored boxes). When generating the schedule, we instead list out all $h$-hop paths from all sources for $h=1,2,...$ (in an ascending order). For hop-1 paths, there can be no contention and all flows can be packed into the same round. For each subsequent $h>1$, we allocate $k^{h-1}$ rounds. Similar to the Circulant case, we pack paths greedily into rounds, considering a next round only when all previous rounds are full or contend with the current path. We then enumerate over each round, select the shortest paths to each destination for each source, break ties randomly, and remove all other paths (\ie the grayed-out paths in Figure~\ref{fig:design_degkgkau_schbase}). At the end of this process, we iterate over the schedule and opportunistically combine rounds that do not contend with each other.

Figures~\ref{fig:design_degkcirc_schbase} and~\ref{fig:design_degkgkau_schbase} shed light on the intuition of how to select between Circulant and GenKautz. Circulant admits rounds that are as fully packed as possible, thus yielding a schedule with fewer rounds.%, due to its strong intra-topology symmetry. However, the weaker expansion of Circulant means that when $n$ gets larger, the paths will also get larger, leading to larger number of hops per round. 
In contrast, GenKautz minimizes longest shortest paths, thanks to its high expansion, but leads to schedules with more "holes" as shown in Figure~\ref{fig:design_degkgkau_schbase}. Specifically, these holes are wasted network capacity and increase the number of rounds required in the schedule. GenKautz is more favorable when $n$ is large and the impact of path length is more evident, or when $d$ (the number of reconfigurations) is smaller and traffic is more concentrated in a few topologies, hence leaving fewer holes.
In our approach, we compute two strategies, each corresponding to either Circulant or GenKautz, and pick the one that yields lower cost. Our empirical results align with the intuition above that when $n$ is larger and $d$ is smaller, GenKautz is favored; otherwise, Circulant is favored.
\vspace{-0.1in}

\begin{figure}[tp]
  \centering

  \begin{subfigure}{\linewidth}
    \centering
    \includegraphics[width=0.35\linewidth]{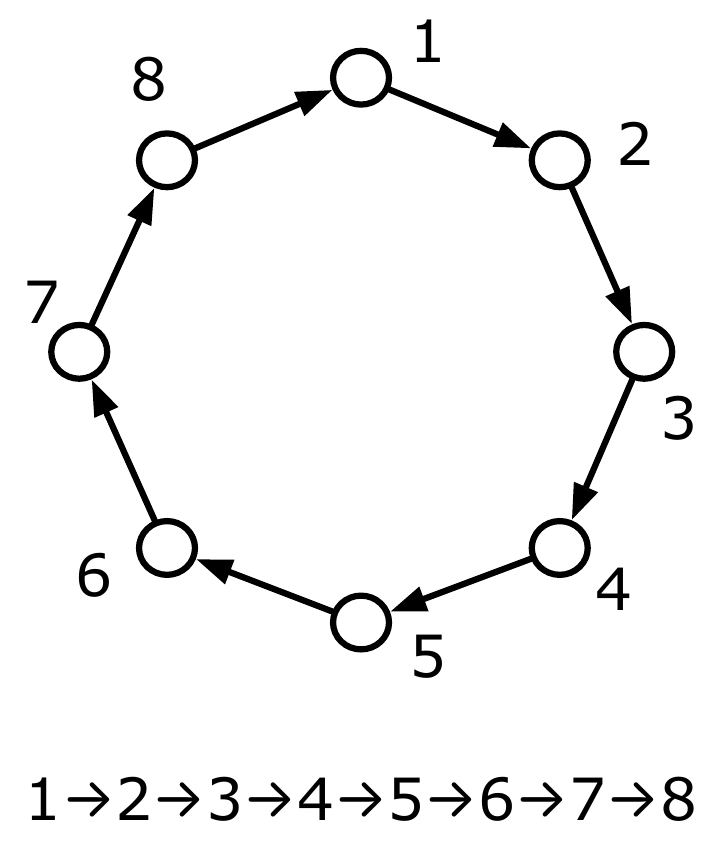}
    \vspace{-0.05in}
    \caption{Non-shifted topology.}
    \label{fig:design_countertopo}
  \end{subfigure}

  % \vspace{0.6em} % adjust vertical gap

  \begin{subfigure}{\linewidth}
    \centering
    \includegraphics[width=0.98\linewidth]{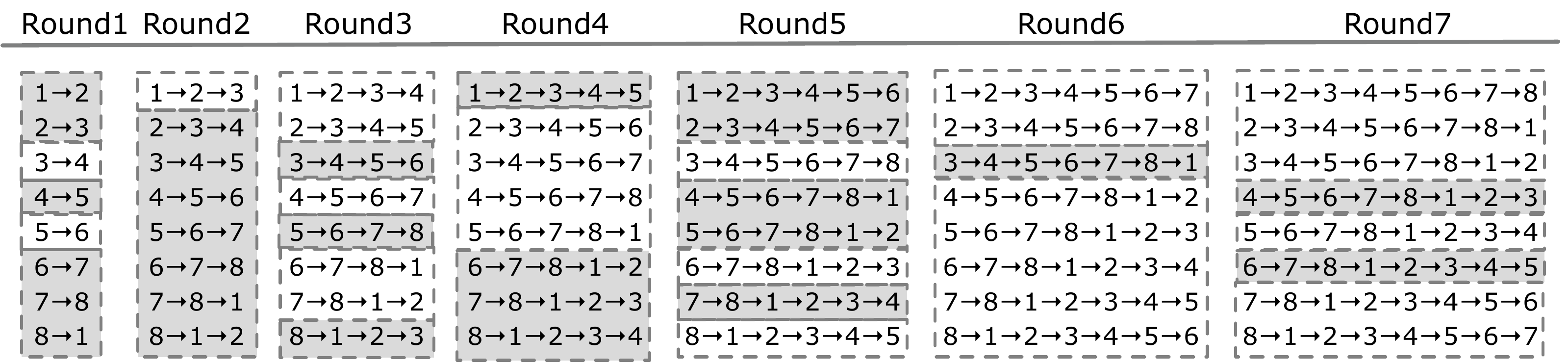}
    \vspace{-0.05in}
    \caption{Schedule for the non-shifted topology.}
    \label{fig:design_countersch}
  \end{subfigure}
    \vspace{-0.25in}
  \caption{A counter-example where a non-shifted topology leads to poor inter-topology symmetry, and thus high cost, for $k=1$.}
  \label{fig:design_counter}
  \vspace{-0.25in}
\end{figure}

\subsection{Topology Sequence Generation \& Scheduling across Sequence} \label{sec:design_toposeq}

\myitem{Degree-1.} 
We generate a sequence of $d-1$ topologies by "shifting" on the base topology.
Following the example in Figure~\ref{fig:design_deg1topo}, Figure~\ref{fig:design_deg1toposhift} shows a 3-shifted topology based on the base topology. Specifically, with the base $\sigma$-sequence, $Q^b = \langle \sigma^b_1, \sigma^b_2, \sigma^b_3, ... \sigma^b_n\rangle$, node $\sigma^b_i$ is connected to node $\sigma^b_{i+1}$. Now we construct a topology where node $\sigma^b_i$ is connected to node $\sigma^b_{i+3}$ instead, thus a 3-shifted topology. In Figure~\ref{fig:design_deg1topo}, we colorcode three of the connections and demonstrate how the 3-shifts happen.
When both topologies are available, our schedule prefers the shortest path. Intra-topology symmetry essentially allows us to schedule just for a reference node and have all other nodes follow the reference node.

First, this generation ensures high expansion across the sequence, as each node will be connected to different nodes for any shift $>1$.
More critically, this generation also ensures inter-topology symmetry. Figure~\ref{fig:design_deg1sch} illustrates the two schedules corresponding to the base and 3-shifted topology, which we will refer to for intuition behind the inter-topology symmetry. Rounds are boxed and colored, and a round with a filled box means it is kept in the final schedule, while a round with a transparent box means it is removed. Take the two purple boxes as an example, which contain the exact set of flows but with different paths. Our schedule prefers shorter paths, thus the entire Round3 in the base schedule is removed, while the entire Round1 in the shifted schedule is kept. After accounting for all rounds, we are left with a schedule, across the two topologies, where Round1,2,4,5 from the base schedule and Round1,2,5 from the shifted schedule are kept, leading to a total transmission cost of $19T$ and thus a total cost of $19T+2R$. The fact that we can remove or retain rounds in full points to the inter-topology symmetry of this generated topology sequence. In addition, in our algorithm, we prefer shifts that still maintain the $n$-cycle rather than breaking into smaller-length cycles, but empirically, we find that the order of shifting is not critical as long as we prioritize using a $n-1$-shift.

Now, if we consider a counter-example illustrated in Figure~\ref{fig:design_counter}, where a second topology is constructed in Figure~\ref{fig:design_countertopo}, in addition to the base topology in Figure~\ref{fig:design_deg1topobase}, but not according to the algorithm we describe above. Figure~\ref{fig:design_countersch} shows the resulting schedule partially \ie which flows from the second schedule are kept in the final schedule (which flows from the first schedule are kept is omitted), where the flows in gray boxes are retained and the flows in white boxes are removed. It is obvious from the figure that even though we have reconfigured an additional topology and some flows can traverse a shorter path, the lack of inter-topology symmetry means that none of the rounds can be removed entirely. We could also opportunistically combine some of the rounds to reduce the total cost, but doing so will not generate a schedule with a cost lower than our shifting and will increase the computational overhead.

\myitem{Degree-k.} 
We generate a sequence of $d-1$ topologies by "contracting" based on the base topology and schedule, for both Circulant and GenKautz. The intuition here is that when a new topology becomes available, we should \emph{at least} be able to reduce by one round. We do so by directly connecting each flow's source and destination in one selected round in the previous schedule. Doing so generates a valid topology that respects $k$. We are not guaranteed that each node will reach $k$ new destinations, thus having weaker inter-topology expansion, but we are guaranteed that the cost will reduce at least by the number of hops in the contracted round. The colored boxes in Figures~\ref{fig:design_degkcirc_schbase} and~\ref{fig:design_degkgkau_schbase} indicate the rounds contracted. In our algorithm, we pick the round to contract by the number of hops in the previous schedule.
In addition, notice that we can also name the "contraction" based on the base Circulant topology as "shifting", \ie "shifting" is a special case of "contracting" that happens when the base topology has perfect intra-topology symmetry.
\vspace{-0.1in}

\subsection{Optimization for Varied-Sized Flows}  \label{sec:design_opt}

When flows are of different sizes, we implement an optimization that relabels the nodes such that node pairs that send larger flows are grouped into the same round as much as possible. 
In addition, note that since our approach leverages the symmetric nature of all-to-all, our performance deteriorates when the input collective is \emph{very skewed}.
\vspace{-0.05in}
\section{Evaluation}

Our evaluation focuses on demonstrating the following:

(1) We reduce the total collective completion time compared to state-of-the-art topology~\cite{liao2025mixnet}, scheduler~\cite{lei2025fast} and optimization frameworks~\cite{zhao2025direct-connect,wang2023topoopt} across a wide range of network parameters, data sizes and workload types (\S\ref{sec:eval_main}).

(2) Our generated topology and schedule sequences are within a small gap of the lower bound (\S\ref{sec:eval_optimality}).

(3) Per-reconfiguration cost affects our reconfiguration decisions (\S\ref{sec:eval_reconfig}).

% (4) Our schedule generation is fast.

\vspace{-0.1in}

\subsection{Methodology}
% \vspace{-0.1in}

\myitem{Setup.}
We evaluate using the htsim packet-level simulator~\cite{htsim}. Reconfigurations are implemented as changes in the input topologies and available routes, and reconfigurations only happen when all current flows have finished. Forwarding is implemented as store-and-forward, and flows with multiple chunks are simulated by indexing the chunks and requiring all nodes to prioritize sending chunks of smaller indices.
We set the link bandwidth to 800Gbps and link latency to 500ns, following prior work~\cite{adaptivephotonicshotnets25vamsi,khan2022linklat500ns}. We set the chunk size to 4MB.
We evaluate four sizes of the scale-up network by varying $n$, the number of GPUs per server, over $\{8,16,32,64\}$. We also vary $k$, the number of optical switches per server or the number of outgoing links per GPU, over $\{1,2\}$. 
We evaluate three types of workloads: (1) Uniform: all GPU pairs send flows of the same size. (2) Random: flow sizes among GPU pairs are generated randomly from a uniform distribution. (3) Zipfian: flow sizes among GPU pairs are generated from a Zipfian distribution with $factor=0.4$. The latter workload is more skewed, with many mice flows and a couple of elephant flows. Prior work~\cite{lei2025fast} reports that their profiling of MoE pretraining follows Zipfian distribution.
For uniform-sized workload, we vary the size of flows for each pair of GPUs over 8MB,16MB,32MB,64MB. For random-sized and Zipfian-sized workloads, we vary the average flow size over 8MB,16MB,32MB,64MB.

\myitem{Baselines.}
We compare against three classes of baselines:
(1) Scale-up topology design: MixNet~\cite{liao2025mixnet} greedily constructs a topology based on the given all-to-all traffic matrix. We do not show MixNet for the uniform workload, since all flows are of the same size and it is unclear how MixNet breaks ties.
(2) Scheduler: FAST~\cite{lei2025fast} is the state-of-the-art scheduler for all-to-all GPU communication, which relies on BvN decomposition and generates schedules faster than solver-based schedulers \eg TACCL~\cite{shah2023taccl}, TE-CCL~\cite{liu2024teccl}, SyCCL~\cite{cao2025syccl}. Note that FAST originally targets a two-hierarchical network comprising both the scale-up and scale-out.
(3) Optimization frameworks: DirectConnect~\cite{zhao2025direct-connect} proposes both Generalized Kautz and Circulant graphs with a multi-commodity flow routing formulation solved by linear programs. TopoOpt~\cite{wang2023topoopt} proposes union of shifted rings with $k$-shortest paths (for the model-parallel traffic that are closer to our workloads than all-reduce, which involves data dependency). We set $k=1$.
We apply our exact or close approximate scheduling when the baseline itself does not include a schedule.
We refer to our strategy as \textsf{ReconfigureDuringCollective}.
% \vspace{-0.2in}

\begin{figure*}[tp]
\centering
\includegraphics[width=0.98\textwidth]{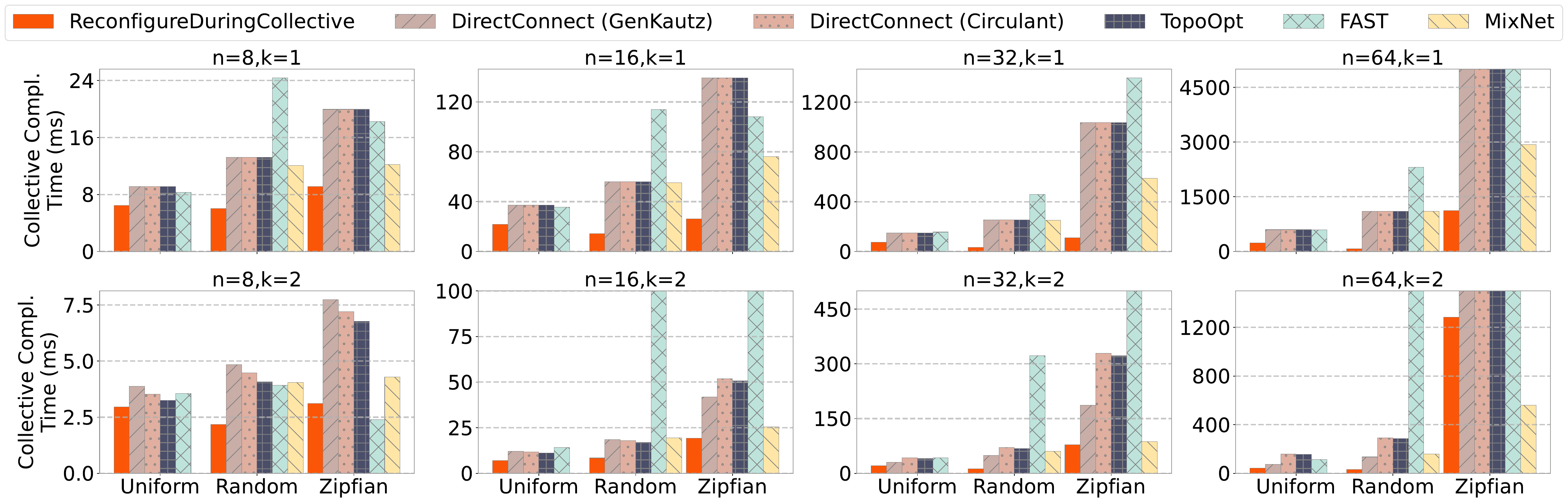}
\vspace{-0.15in}
\caption{Reconfiguring during the collective substantially reduces collective completion time across a wide range of scale-up configurations and workloads. \textsf{ReconfigureDuringCollective} performs the best for uniform and random traffic and remain competitive under skewed Zipfian traffic.}
\label{fig:eval_main}
\vspace{-0.15in}
\end{figure*}

\vspace{-0.05in}
\subsection{Performance} \label{sec:eval_main}
\vspace{-0.05in}

\myitem{We outperform baselines across a wide range of scale-up configurations under different workloads.}
Figure~\ref{fig:eval_main} reports the total collective completion time, computed as the sum of the total reconfiguration time and the data transmission time, across eight scale-up configurations with different numbers of GPUs ($n$) and optical switches ($k$) under three workloads. For illustrative purposes, we truncate exceptionally tall bars; truncated bars therefore appear to hit the top of a subfigure.
This figure is intended to show the maximum achievable benefit of allowing reconfiguration during the collective. For each configuration, we generate the best strategy for every candidate number of reconfigurations, compute its total cost, and then sweep a range of per-reconfiguration costs. We experiment with the strategy that attains the largest cost benefit over this sweep, and we use the same per-reconfiguration cost values when evaluating baselines.
For this plot, we fix flow sizes to 32MB. For random and Zipfian workloads, we have applied our relabeling optimization.
Across all eight scale-ups and three workloads, we achieve 39.66\% lower collective completion time on average relative to the best baseline in each setting. The gains generally increase as the $n/k$ ratio grows.
Breaking down by workload, we perform best under Uniform and Random traffic, and remain competitive under Zipfian traffic, though the improvement is smaller. This is expected: Zipfian demand is more skewed, which breaks the symmetry inherent to all-to-all communication that our design leverages.
MixNet, a greedy scheme that better matches skew, indeed outperforms us in one case ($n{=}64$, $k{=}2$, Zipfian). FAST, which depends on a BvN decomposition, exhibits high variance: it can perform very poorly in many settings, particularly at larger scales, but it also surpasses us in one case ($n{=}8$, $k{=}2$, Zipfian).
Finally, when $k{=}1$, DirectConnect (GenKautz), DirectConnect (Circulant), and TopoOpt coincide because the network effectively reduces to a single directed ring (an $n$-cycle). In this case, the linear-program schedule matches shortest-path routing. More broadly, this also points to the fact that without reconfiguration during the collective, there is limited room to optimize, and the gap from their results to our results directly demonstrates the power of enabling reconfiguration during collectives.

\begin{figure*}[tp]
\centering
\includegraphics[width=0.98\textwidth]{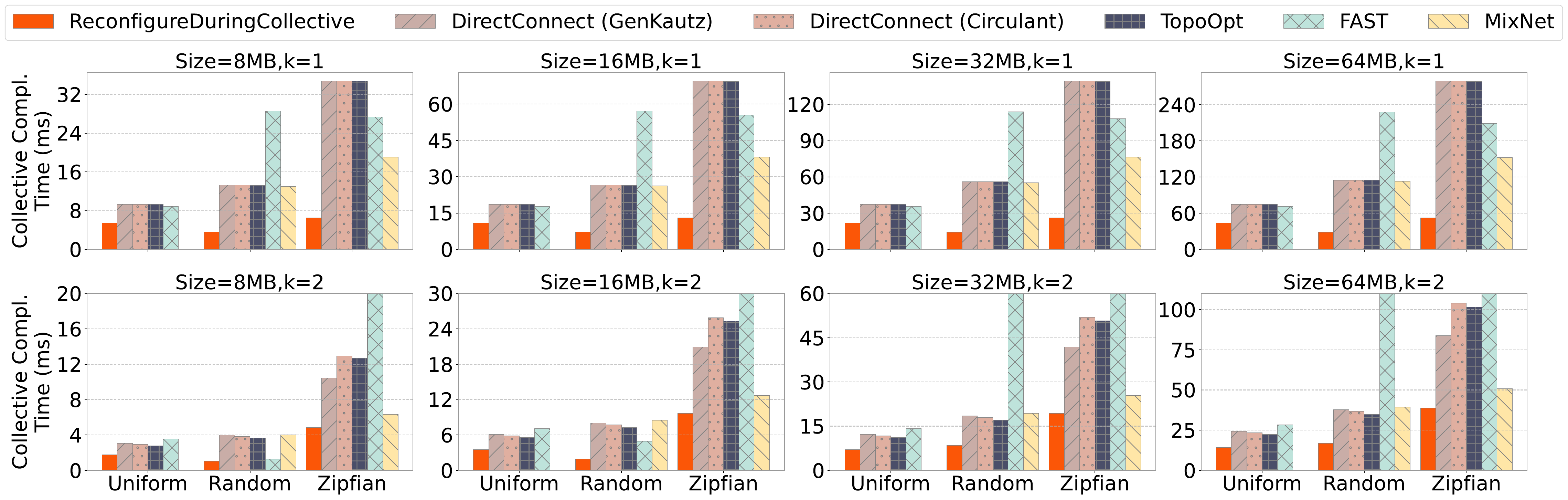}
\vspace{-0.15in}
\caption{Reconfiguring during the collective consistently reduce collective completion time across flow sizes, because contention-free scheduling prevents larger flows from creating additional contention.}
\label{fig:eval_main_size}
\vspace{-0.1in}
\end{figure*}

\myitem{We outperform baselines across different flow sizes under different workloads.}
Figure~\ref{fig:eval_main_size} reports the total collective completion time across a range of flow sizes under three workloads, for $k=1$ and $k=2$ with $n=16$ (we observe the same trend for other $n$). As in Figure~\ref{fig:eval_main}, we report the maximum attainable gain from reconfiguring during collectives (\ie sweeping the per-reconfiguration cost) and truncate exceptionally high bars for illustrative purposes. Across all flow sizes, our approach consistently outperforms the baselines, yielding an average 47.35\% reduction in completion time. The improvement persists as flows grow because our contention-free schedules prevent additional offered load from translating into additional contention.

% \vspace{-0.2in}

\vspace{-0.15in}
\subsection{Optimality Gap} \label{sec:eval_optimality}
\vspace{-0.05in}

\begin{figure}[tp]
  \centering

  \begin{subfigure}{\linewidth}
    \centering
    \includegraphics[width=0.5\linewidth]{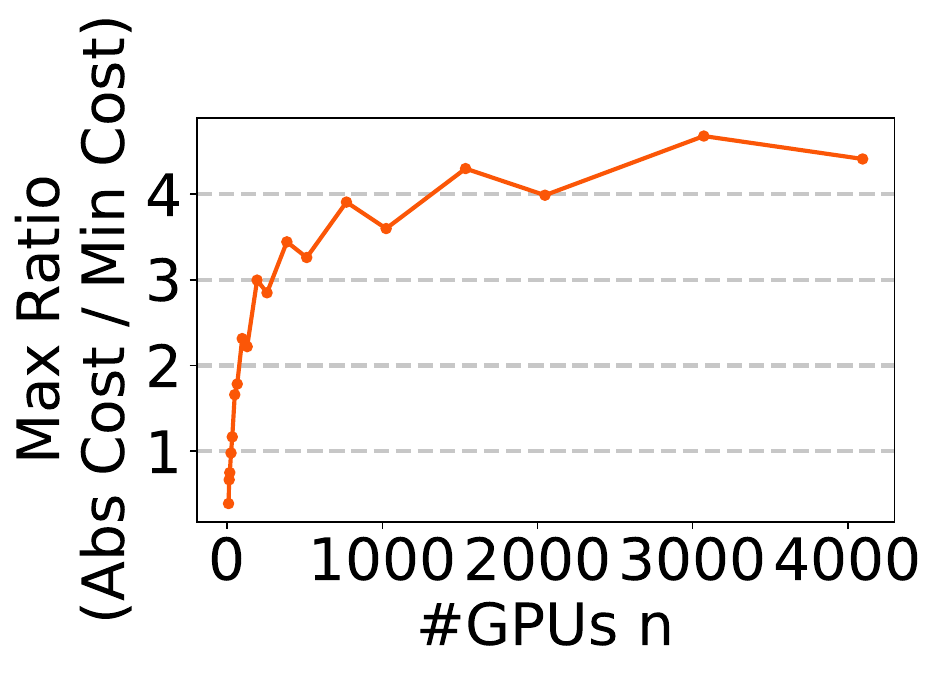}
    \vspace{-0.1in}
    \caption{The worst-case gap between our generated strategies and the lower bound.}
    \label{fig:eval_optimality_ratio}
  \end{subfigure}

  % \vspace{0.6em} % adjust vertical gap

  \begin{subfigure}{\linewidth}
    \centering
    \includegraphics[width=0.98\linewidth]{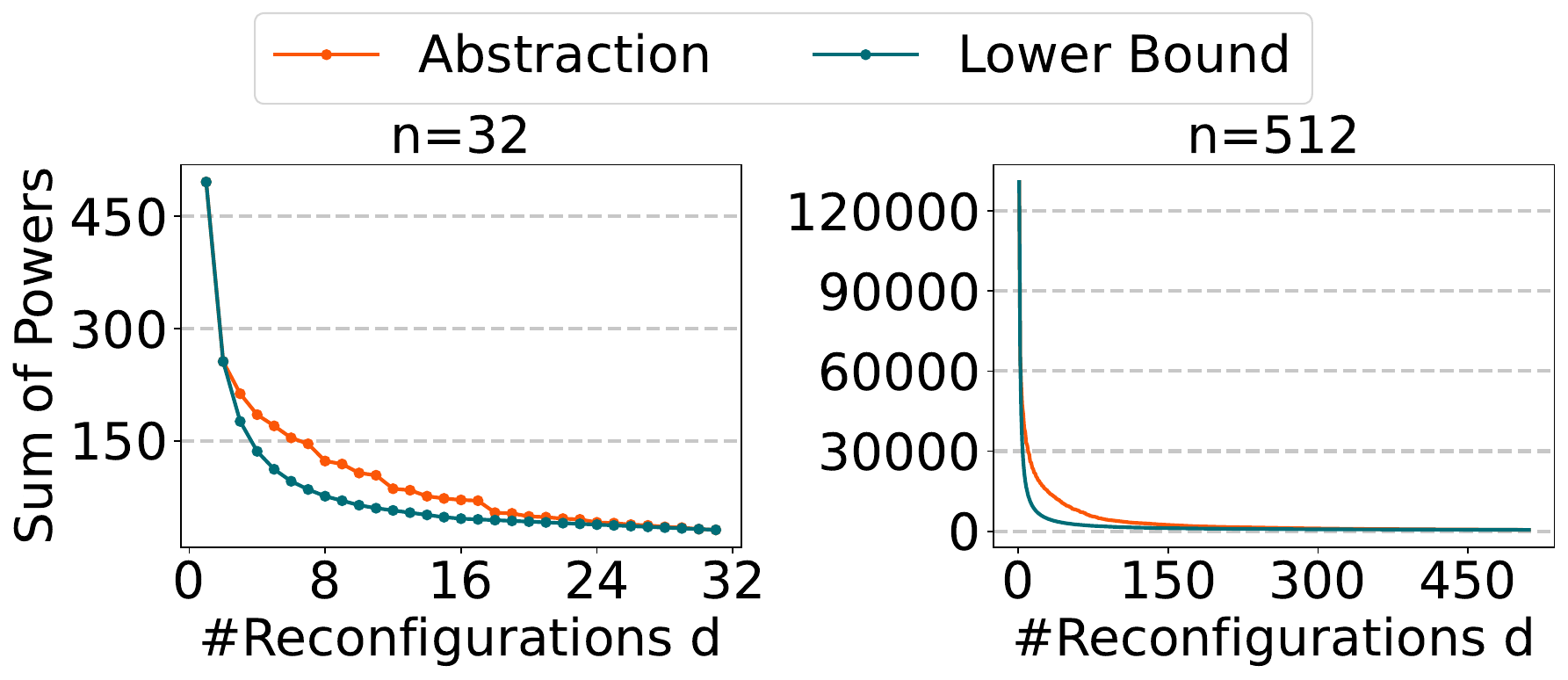}
    \vspace{-0.1in}
    \caption{Examples of how the cost of our generated strategies compares to the lower bound across all candidate $d$ values.}
    \label{fig:eval_optimality_powersum}
  \end{subfigure}
\vspace{-0.15in}
  \caption{Strategies generated based on the selected region stay within a small constant factor of the lower bound for the global optimum ($k$=1). This is illustrated by the small max ratio values in (\ref{fig:eval_optimality_ratio}) and the small gap between the two lines in (\ref{fig:eval_optimality_powersum}).}
  \label{fig:eval_optimality}
  \vspace{-0.2in}
\end{figure}

\myitem{Topology and schedule sequences generated based on the region with highly symmetrical and expandable topologies stay within a small factor of the global-optimum lower bound.}
Figure \ref{fig:eval_optimality} evaluates  how close the topology and schedule sequences generated based on our selected region, characterized by highly symmetrical and expandable topologies, come to the lower bound for the global optimum \ie Theorem~\ref{thm:lowerbound}. Figure ~\ref{fig:eval_optimality_ratio} summarizes worst-case optimality gaps: for each $n$ value, we compute the maximum ratio between our power sum and the lower bound across all candidate number of reconfigurations ($d$), and we plot this maximum ratio across $n$. Figure~\ref{fig:eval_optimality_powersum} then illustrates two specific examples ($n=32$ and $n=512$, randomly selected), showing our power sum and the lower bound across all $d$ for each $n$.
Two takeaways emerge. 
First, in Figure~\ref{fig:eval_optimality_ratio}, the worst-case gap grows very slowly with scale $n$. Up to $n=4096$, the maximum ratio remains $\leq 4.54$ (4096 GPUs inside the same server is what Google uses in its TPUv4~\cite{jouppi2023tpu}); at practical sizes $n \le 64$, it is at most 2.22. In other words, sequences generated by our abstraction are within a small constant factor of the lower bound—and this bound is not necessarily achievable in every problem instance—highlighting the effectiveness of our selected region.
Second, Figure~\ref{fig:eval_optimality_powersum} shows that the lower bound decreases sharply for small $d$ and then flattens as $d$ increases. This matches Theorem~\ref{thm:lowerbound}, which predicts diminishing returns from additional reconfigurations when $d$ gets relatively large and each additional reconfiguration saves the cost of one hop. Practically, this emphasizes the importance of performing well at small $d$: when only a few reconfigurations are allowed (the most constrained setup), our generated sequences should still capture most of the available benefit, ensuring reconfiguration remains worthwhile even when $d$ is small.
% \vspace{-0.15in}

\vspace{-0.15in}
\subsection{Impact of Per-Reconfiguration Cost} \label{sec:eval_reconfig}
\vspace{-0.05in}

\begin{figure}[tp]
  \centering

  \begin{subfigure}{\linewidth}
    \centering
    \includegraphics[width=0.8\linewidth]{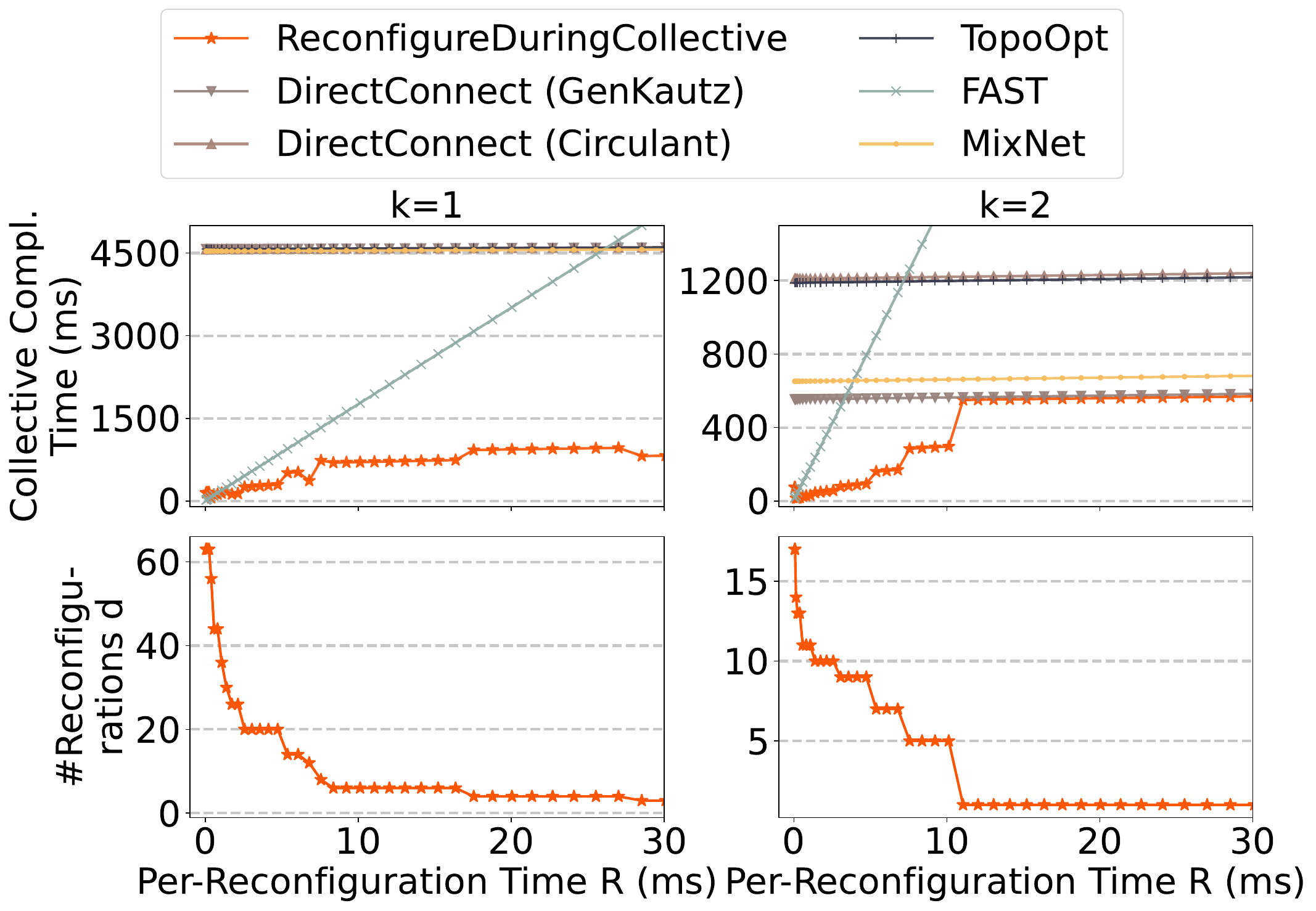}
    \vspace{-0.1in}
    \caption{$n=64$, size=128MB}
    \label{fig:eval_reconfig_n64_s128}
  \end{subfigure}

  % \vspace{0.6em} % adjust vertical gap

  \begin{subfigure}{\linewidth}
    \centering
    \includegraphics[width=0.8\linewidth]{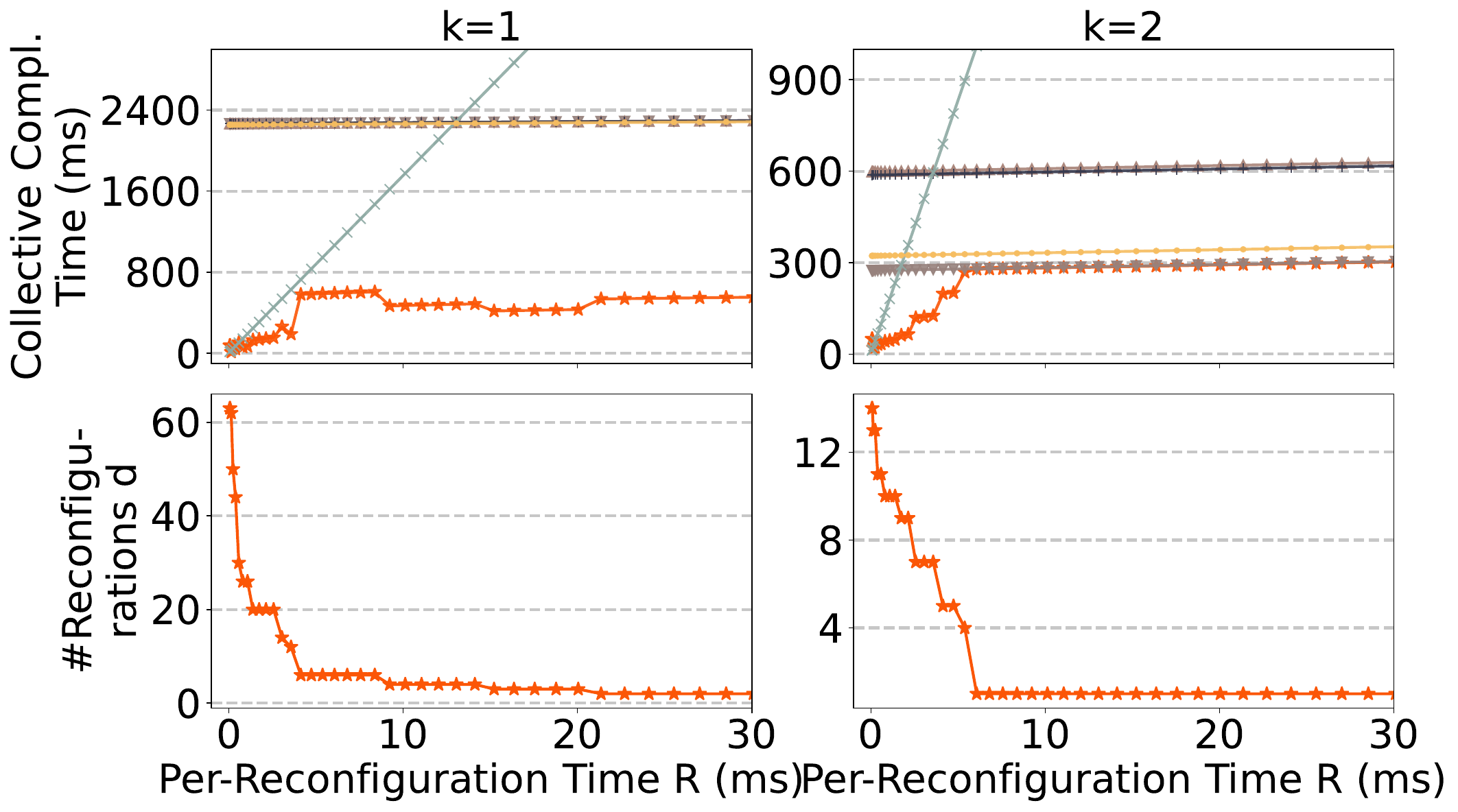}
    \vspace{-0.1in}
    \caption{$n=64$, size=64MB}
    \label{fig:eval_reconfig_n64_s64}
  \end{subfigure}

  % \vspace{0.6em} % adjust vertical gap

  \begin{subfigure}{\linewidth}
    \centering
    \includegraphics[width=0.8\linewidth]{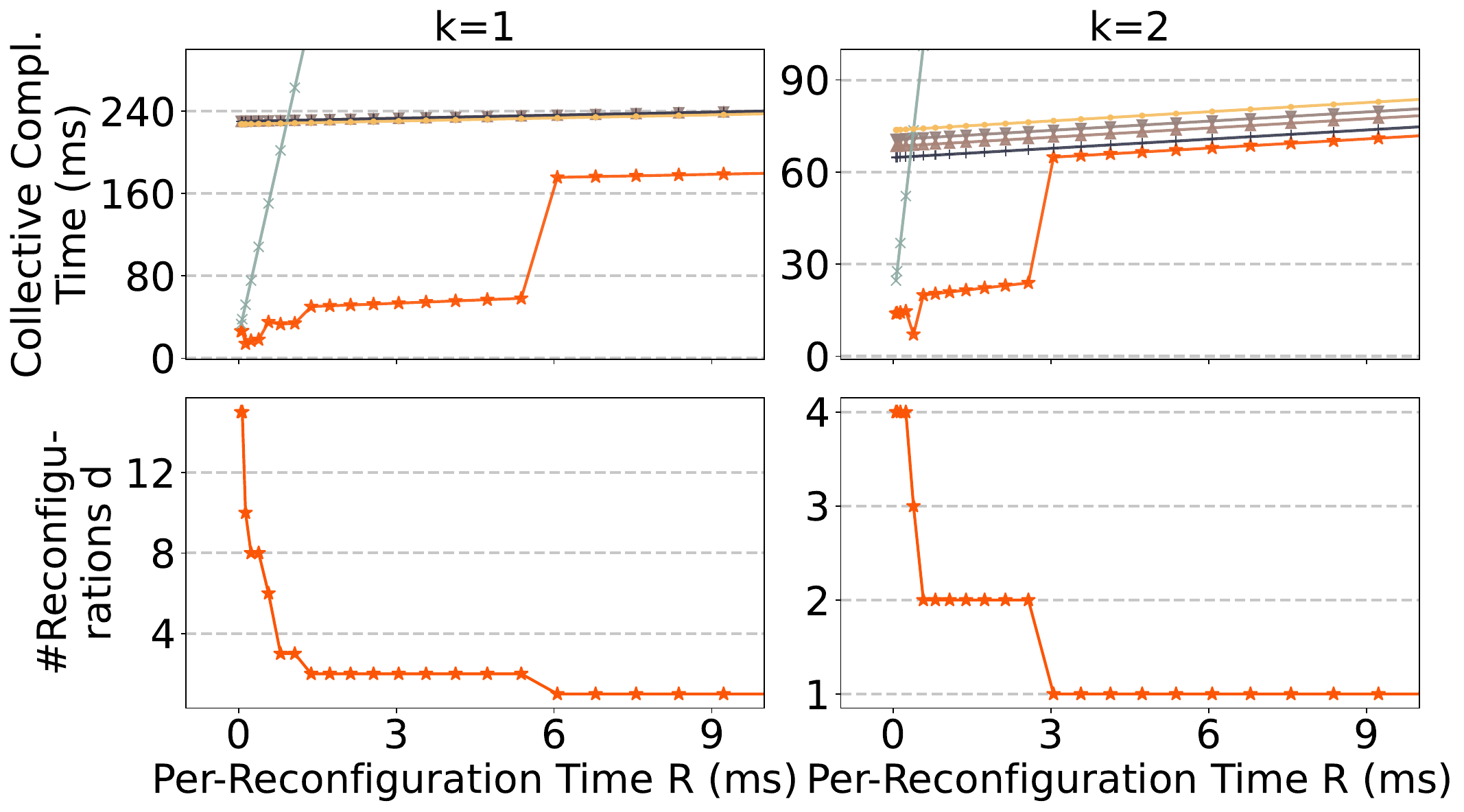}
    \vspace{-0.1in}
    \caption{$n=16$, size=128MB}
    \label{fig:eval_reconfig_n16_s128}
  \end{subfigure}
\vspace{-0.2in}
  \caption{\textsf{ReconfigureDuringCollective} selects the right number of reconfigurations, outperforming baselines over a broad, realistic range of per-reconfiguration costs.}
  \label{fig:eval_reconfig}
  \vspace{-0.2in}
\end{figure}

\myitem{We outperform baselines across a wide range of realistic per-reconfiguration costs by adapting the number of reconfigurations.} 
Figures~\ref{fig:eval_main} and~\ref{fig:eval_main_size} report the best possible improvement (\ie the maximum gain after sweeping the per-reconfiguration delay). Here we instead study the operational setting where the per-reconfiguration cost is given. After generating the best topology/schedule sequence for each candidate reconfiguration count $d$, we use our cost model to select the $d$ that minimizes predicted completion time for the specified $R$. Figure~\ref{fig:eval_reconfig} plots the resulting collective completion time for the random workload, and also provides intuition for how the chosen $d$ changes with $R$.
Figure~\ref{fig:eval_reconfig_n64_s128} fixes $n{=}64$ and flow size 128MB. As $R$ increases, our selected $d$ decreases, as expected. Moreover, the transition to using only a static topology happens earlier when $k$ is larger (\eg $k{=}2$ vs. $k{=}1$). This is because higher $k$ values indicate higher network capacities, which reduces the transmission time, so reconfiguration becomes unprofitable at a smaller $R$. This trend is consistent across the other settings as well.
The curves are not perfectly smooth because the realized transmission time can vary slightly even with chunk scheduling; occasionally, a neighboring $d$ would have been marginally better than the one selected by the model. Even so, we retain a substantial margin over the baselines. Notably, even at $d{=}1$ (no reconfiguration after the initial topology), we can still outperform some baselines due to our better base topology choice.
Comparing Figure~\ref{fig:eval_reconfig_n64_s64} (64MB flows) against the 128MB case illustrates the expected shift: when flow size is smaller, the transmission time is also smaller, and only smaller $R$ values justify reconfiguration, so the “worthwhile” region moves left. The same effect appears when reducing scale (\eg $n{=}16$, shown in Figure~\ref{fig:eval_reconfig_n16_s128}): with a smaller topology, paths are shorter and transmission time is smaller accordingly, so reconfiguration is only beneficial at correspondingly smaller $R$.
Finally, these plots also clarify where different baselines excel. FAST performs best when $R$ is very small, while fixed-topology designs (DirectConnect, TopoOpt and MixNet) tend to be strongest when $R$ is very large, which is consistent with their underlying assumptions. In contrast, by adapting $d$, we achieve strong performance at both extremes and, importantly, maintain a clear advantage in the intermediate range of $R$ that is representative of practical settings.
\section{Conclusion}
This paper presents a unifying matrix-based formulation for all-to-all communication over reconfigurable photonic interconnects that captures the full space of topology sequences and flow schedules. Guided by this abstraction, we identified symmetric, high-expansion topology sequences that, together with a simple flow-splitting heuristic, yield strategies that are near optimal in completion time.
This work does not raise any ethical issues.

\clearpage
\bibliographystyle{ACM-Reference-Format}
\bibliography{reference}

\end{document}